\documentclass[11pt]{article}

\usepackage{fullpage}
\usepackage{microtype}
\usepackage{graphicx}
\usepackage{subfigure}
\usepackage{booktabs} % for professional tables
\usepackage{natbib}

\newcommand{\ignore}[1]{}

\usepackage{amsmath}
\usepackage{amssymb}
\usepackage{mathtools}
\usepackage{amsthm}

\usepackage{url}
\usepackage[algo2e, ruled, vlined]{algorithm2e}
\usepackage{xcolor}

\usepackage{bbm}
\theoremstyle{plain}
\newtheorem{theorem}{Theorem}[section]

\newtheorem{lemma}[theorem]{Lemma}

\theoremstyle{definition}
\newtheorem{definition}[theorem]{Definition}

\theoremstyle{remark}
\newtheorem{remark}[theorem]{Remark}
\newtheorem{claim}[theorem]{Claim}
\newtheorem{problem}[theorem]{Problem}

\newenvironment{proofof}[1]{\begin{proof}[{\textit{Proof of #1}}]}{\end{proof}}

\newcommand{\edith}[1]{{\color{purple} \small Edith says: #1}}

\def\E{\textsf{E}}
\def\Var{\textsf{Var}}

\def\Pr{\textsf{Pr}}
\def\Exp{\textsf{Exp}}
\def\Geom{\textsf{Geom}}
\def\Bern{\textsf{Bern}}
\def\Binom{\textsf{Binom}}

\usepackage[textsize=tiny]{todonotes}

\usepackage{xcolor}
\usepackage[colorlinks=true,linkcolor=blue, citecolor=blue, urlcolor=blue]{hyperref}
\usepackage[algo2e, ruled, vlined]{algorithm2e}

\usepackage{authblk}

\title{Unmasking Vulnerabilities: \\ Cardinality Sketches under Adaptive Inputs}

\author[1]{Sara Ahmadian}
\author[1,2]{Edith Cohen}
\affil[1]{Google Research, United States}
\affil[2]{Department of Computer Science, Tel Aviv University, Israel}
\affil[ ]{\texttt{sahmadian@google.com}}
\affil[ ]{\texttt{edith@cohenwang.com}}

\date{}

\begin{document}
\maketitle

\begin{abstract}
 Cardinality sketches are popular data structures that enhance the efficiency of working with large data sets. The sketches  are randomized representations of sets that are only of logarithmic size but can support set merges and approximate cardinality (i.e., distinct count) queries. When queries are not adaptive, that is, they do not depend on preceding query responses, the design provides strong guarantees of correctly answering a number of queries exponential in the sketch size $k$. 
 In this work, we investigate the performance of cardinality sketches in adaptive settings and unveil inherent vulnerabilities. We design an attack against the ``standard'' estimators that constructs an adversarial input by post-processing responses to a set of simple non-adaptive queries of size linear in the sketch size $k$. Empirically,  our attack used only $4k$ queries with the widely used HyperLogLog (HLL++)~\citep{hyperloglog:2007,hyperloglogpractice:EDBT2013} sketch. The simple attack technique suggests it can be effective with post-processed natural workloads. Finally and importantly, we demonstrate that the vulnerability is inherent as \emph{any} estimator applied to known sketch structures can be attacked using a number of queries that is quadratic in $k$, matching a generic upper bound.
\end{abstract}
\section{Introduction}

Composable sketches for cardinality estimation are data structures that are commonly used in practice~\citep{datasketches,bigquerydocs} and had been studied extensively~\citep{FlajoletMartin85,flajolet2007hyperloglog,ECohen6f,BJKST:random02,KNW:PODS2010,NelsonNW13,ECohenADS:TKDE2015,PettieW:STOC2021}. The sketch of a set is a compact representation that supports merge (set union) operations, adding elements, and retrieval of approximate cardinality. 
The sketch size is only logarithmic or double logarithmic in the cardinality of queries, which allows for a significant efficiency boost over linear size data structures such as Bloom filters.

Formally, for a universe $\mathcal{U}$ of keys, and randomness $\rho$, 
a sketch map $U\mapsto S_\rho(U)$ is a mapping from sets of keys 
$U\in 2^\mathcal{U}$ to their sketch $S_\rho(U)$. Sketch maps are designed so that 
for each $U$ we can recover with high probability (over $\rho$) an estimate of $|U|$ by applying an estimator $\mathcal{M}$ to $S_\rho(U)$. A common guarantee 
is a bound on the Normalized Root Mean Squared Error (NRMSE) so that for accuracy parameter $\epsilon$, 
\begin{equation} \label{nrmse:eq}
\forall U,\ \E_\rho\left[\left( \frac{\mathcal{M}(S_\rho(U)) - |U|}{|U|}\right)^2 \right] \leq \epsilon^2\ .
\end{equation}
The maps $S_\rho$ are designed to be \emph{composable}:  For a set $U$ and key $u$, the sketch $S_\rho(U\cup\{u\})$ can be computed from $S_\rho(U)$ and $u$. For two sets $U$, $V$, the sketch $S_\rho(U\cup V)$ can be computed from their respective sketches $S_\rho(U)$, $S_\rho(V)$. Composability is a crucial property that makes the sketch representations useful for streaming, distributed, and parallel applications. 
Importantly, the use of the same internal randomness $\rho$ across all queries is necessary for composability and therefore in typical use cases  it is fixed across a system.

The basic technique in the design of cardinality (distinct count) sketches is to randomly prioritize the keys in the universe $\mathcal{U}$ through the use of random hash functions (specified by $\rho$). The sketch of a set $U$ keeps the lowest priorities of keys that are in the set $U$.  This provides information on the cardinality $|U|$, since a larger cardinality corresponds to the presence of lower priorities keys in $U$.
This technique was introduced by~\citet{FlajoletMartin85} for counting distinct elements in streaming and as composable sketches of sets by~\citet{ECohen6f}. The core idea of sampling keys based on a random order emerged in reservoir sampling \citep{Knuth2}, and in weighted sampling ~\citep{Rosen1997a}. Cardinality sketches are also Locality Sensitive Hashing (LSH) maps~\citep{IndykMotwani:stoc98} with respect to set differences.

The specific designs of cardinality sketches vary and include \emph{MinHash sketches} (randomly map keys to priorities) or \emph{domain sampling} (randomly map keys to sampling rates).
With these methods, 
the sketch size dependence on the maximum query size $|U|\leq n$ is $\log n$ or $\log\log n$. The sketch size (number of registers) needed for the NRMSE guarantee of Equation~\eqref{nrmse:eq} is $k=O(\epsilon^{-2})$. 
The sketch size needed for the following $(\epsilon,\delta)$ guarantee (confidence $1-\delta$ of relative error of $\epsilon$)  
\begin{equation} \label{accconf:eq}
\forall U,\ \Pr_\rho\left[\left| \frac{\mathcal{M}(S_\rho(U)) - |U|}{|U|}\right| > \epsilon \right] \leq \delta\ 
\end{equation}
is $k= O(\epsilon^{-2}\log(1/\delta))$.  
This guarantee means that for any $U$, almost any sampled $\rho$ works well.

We model the use of the sketching map $S_\rho$ as an interaction between a \emph{source}, that issues queries $U_i \subset \mathcal{U}$, and a \emph{query response (QR)} algorithm, that receives the sketch 
$S_\rho(U_i)$ (but not $U_i$) applies an \emph{estimator} $\mathcal{M}$ and returns the estimate
$\mathcal{M}(S_\rho(U_i))$
on the cardinality $|U_i|$.

When randomized data structures or algorithms are invoked interactively, it is important to make a distinction between non-adaptive queries, that do not depend on $\rho$, and adaptive queries. 
In non-adaptive settings we can treat queries as fixed in advance. In this case, we can apply a union bound with the guarantee of Equation~\ref{accconf:eq} and obtain that the probability that the responses for all $r$ queries are within a relative error of  $\epsilon$, is at least $1-r\delta$. Therefore, the sketch size needed to provide an $(\epsilon,\delta)$-guarantee on this $\ell_\infty$ error is $k=O(\epsilon^{-2}\log(r/\delta)$. In particular, the query response algorithm can be correct on a number of non-adaptive queries that is exponential in the sketch size until a set is encountered for which the estimate is off. 

Many settings, however, such as control loops, optimization processes, or malicious behavior, give rise to adaptive inputs. This can happen inadvertently when a platform such as~\citet{datasketches} or SQL~\citep{bigquerydocs} is used. 
In such cases, information on the randomness $\rho$ may leak from query responses, and the union bound argument does not hold.
An important question that arises is thus to understand the actual vulnerability of our specific algorithms in such settings.  Are they practically robust? How efficiently can they be attacked?  What can be the consequences of finding such an adversarial input?

Randomized data structures designed for non-adaptive queries can be applied in generic ways with adaptive queries.
However, the guarantees provided by the resulting (robust) algorithms tend to be significantly weaker than those of their non-robust counterparts. 
The straightforward approach is to maintain multiple copies of the sketch maps (with independent randomness) and discard a copy after it is used once to respond to a query. This results in a linear dependence of the number of queries $r$ in the size of the data structure. \citet{HassidimKMMS20} proposed the robustness wrapper method that allows for $r^2$ adaptive queries using $\tilde{O}(r)$ sketch maps. The method uses differential privacy to protect the randomness and the analysis uses generalization~\citep{DworkFHPRR:STOC2015,BassilyNSSSU:sicomp2021} and advanced composition~\citep{DMNS06}.  The quadratic relation is known to be tight in the worst-case for adaptive statistical queries with an attack~\citep{HardtUllman:FOCS2014,SteinkeUllman:COLT2015} designed using Fingerprinting Codes~\citep{BonehShaw_fingerprinting:1998}. But these attacks do not preclude a tailored design with a better utility guarantee for cardinality sketches and also do not apply with ``natural'' workloads.

\subsection*{Contributions and Overview}

We consider the known sublinear composable sketch structures for cardinality estimation, which we review in Section~\ref{prelim:sec}. Our primary contribution is designing attacks that construct a set $U$ that is \emph{adversarial} for the randomness $\rho$. We make this precise in the sequel, but for now, an adversarial set $U$ results in cardinality estimates that are off.

\begin{itemize}
\item
We consider query response algorithms that use the ``standard'' cardinality estimators. These estimators optimally use the information in the sketch and report a value that is a function of a sufficient statistic of the cardinality. 
In Section~\ref{standardattack:sec} we present an attack on these estimators. The product of the attack is an adversarial set, one for which the sketch $S_\rho(U)$ is grossly out of distribution. The attack uses linearly many queries $O(k)$ in the sketch size and importantly, issues all queries in a single batch. The only adaptive component is the post processing of the query responses. This single-batch 
attack suggests that it is possible to construct an adversarial input by simply observing and post-processing a normal and non-adaptive workload of the system. The linear size of the attack matches the straightforward upper bound of using disjoint components of the sketch for different queries.

\item We conduct an empirical evaluation of our proposed attack on the HyperLogLog (HLL) sketch~\citep{DurandF:ESA03,flajolet2007hyperloglog} with the HLL++ estimator~\citep{hyperloglogpractice:EDBT2013}. This is the most widely utilized sketch for cardinality estimation in practice.  The results reported in Section~\ref{experiments:sec} show that even with a single-batch attack using $4k$ queries, we can consistently construct adversarial inputs on which the estimator  substantially overestimates or underestimates the cardinality by 40\%.

\item
In Section~\ref{setupattack:sec} and Section~\ref{singlebatch:sec}, we present an attack that broadly applies against \emph{any} correct query response algorithm. By that, we establish inherent vulnerability of the sketch structures themselves.
Our attack uses $\tilde{O}(k^2)$ \emph{adaptive} queries. We show that multiple batches are necessary against strategic query response algorithms.  
This quadratic attack size matches the generic quadratic upper bound construction of~\citet{HassidimKMMS20}. 
The product of our attack is a small \emph{mask} set $M$ that
can poison larger sets $U$ in the sense that 
$S(M\cup U)\approx S(M)$, making any estimator ineffective.
The attack applies even when the query response is for the  more specialized \emph{soft threshold} problem: Determine if the cardinality is below or above a range of the form $[A,2A]$. Moreover, it applies even when the response is tailored to the attack algorithm and its internal state including the distribution  from which the query sets are selected at each step. Note that this strengthening of the query response and simplification of the task only makes the query response algorithm harder to attack. 
\end{itemize}

Our attacks have the following structure: We fix a ground set $N$ of keys and issue queries that are subsets $U_i\subset N$. We maintain scores to keys in $N$ that are adjusted for the keys in $U_i$ based on the response. The design has the property that scores are correlated with the priorities of keys and the score is higher when the cardinality is underestimated. The  adversarial set is then identified as a prefix or a suffix of keys ordered by their score.

The vulnerabilities we exposed may have practical significance in multiple scenarios: In a non-malicious setting, an adaptive algorithm or an optimization process that is applied in sketch space can select keys that tend to be in overestimated (or underestimates) sets, essentially emulating an attack and inadvertently selecting a biased set on which the estimate is off. In malicious settings, the construction of an adversarial input set $U$ can be an end goal. For example, a system that collects statistics on network traffic can be tricked to report that  traffic is much larger or much smaller than it actually is.  A malicious player can poison the dataset by injecting a small adversarial set $M$ to the data $U$, for example, by issuing respective search queries to a system that sketches sets of search queries. The sketch $S_\rho(M\cup U)$ then masks $S_\rho(U)$, making it impossible to recover an estimate of the true cardinality of $U$. 
Finally, cardinality sketches have weighted extensions (max-distinct statistics) and are building blocks of sketches designed for a large class of concave sublinear frequency statistics, that include cap statistics and frequency moments with $p\leq 1$~\citep{CapSampling,CohenGeri:NeurIPS2019,JayaramWoodruff:TALG2023}, and thus these vulnerabilities apply to these extensions.

\section{Related Work}

There are prolific lines of research on the effect of adaptive inputs that span multiple areas including dynamic graph algorithms~\citep{ShiloachEven:JACM1981,AhnGM:SODA2012,gawrychowskiMW:ICALP2020,GutenbergPW:SODA2020,Wajc:STOC2020, BKMNSS22},  sketching and streaming algorithms~\citep{MironovNS:STOC2008,HardtW:STOC2013,BenEliezerJWY21,HassidimKMMS20,WoodruffZ21,AttiasCSS21,BEO21,DBLP:conf/icml/CohenLNSSS22,TrickingHashingTrick:arxiv2022}, adaptive data analysis~\citep{Freedman:1983,Ioannidis:2005,FreedmanParadox:2009,HardtUllman:FOCS2014,DworkFHPRR15} and machine learning~\citep{szegedy2013intriguing,goodfellow2014explaining,athalye2018synthesizing,papernot2017practical}.

\citet{DBLP:journals/icl/ReviriegoT20} and~\citet{cryptoeprint:2021/1139} proposed attacks on the HLL sketch with standards estimators.  The proposed attacks were in a streaming setting and utilized many dependent queries in order to detect keys whose  insertion results in updates to the cardinality estimate. Our attacks are more general: We construct single-batch attacks with standard estimators and also construct attacks on these cardinality sketches that apply with any estimator.
The question of robustness of cardinality sketches to adaptive inputs is related but different than the well studied question of whether they are differentialy private. 
Cardinality sketches were shown to be not privacy preserving when the sketch randomness or content are public~\citep{DLB:PET2019}.  Other works~\citep{NEURIPS2020_e3019767,paghS:ICDT2021,SteinkKnop:2023}  showed that cardinality sketches are privacy preserving, but this is under the assumption that the randomness is used once.
Our contribution here is designing attacks and 
quantifying their efficiency. The common grounds with privacy is the high sensitivity of the sketch maps to removal or insertion of low priority key.

Several works constructed attacks on linear sketches, including the Johnson Lindenstrauss Transform~\citep{DBLP:conf/nips/CherapanamjeriN20}, the AMS sketch~\citep{BenEliezerJWY21,TrickingHashingTrick:arxiv2022}, and CountSketch~\citep{DBLP:conf/icml/CohenLNSSS22,TrickingHashingTrick:arxiv2022}. The latter showed that the standard estimators for CountSketch and the AMS sketch can be compromised with a linear number of queries and the sketches with arbitrary estimators can be compromised with a quadratic number of queries. The method was to combine low bias inputs with disjoint supports and have the bias amplified, since the bias increases linearly whereas the $\ell_2$ norms increases proportionally to $\sqrt{r}$. This approach does not work with cardinality sketches, which required a different attack structure. Combining disjoint sets on which the estimate is slightly biased up will not amplify the bias.  The common ground, perhaps surprisingly, is that these fundamental and popular sketches are all vulnerable with adaptive inputs and in a similar manner: Estimators that optimally use the sketch require linear size attacks.  Arbitrary correct estimators require quadratic size attacks.

\section{Preliminaries} \label{prelim:sec}

An \emph{attack} is an interaction designed to construct a set that is \emph{adversarial} to the randomness $\rho$. 
An adversarial set can be identified by trying out a large number of inputs.
We measure the efficiency of the attack by its \emph{size} (number of issued queries) and \emph{concurrency} (number of batches of concurrent queries). 

A set $U$ is 
\emph{adversarial} for the randomness $\rho$ if 
a \emph{sufficient statistics} for the cardinality that is computed from $S_\rho(U)$ is very skewed with respect to its distribution under sampling of $\rho$. That is, it has proportionally too few or two many low priority keys.

\begin{definition} [Sufficient Statistics]
A statistic $T$ on the sketch domain $S \mapsto \mathbb{R}$ is \emph{sufficient} for the cardinality $|U|$ if it includes all information the sketch provides on the cardinality $|U|$. That is, for each $t$, the conditional distribution of the random variable $S_\rho(U)$ given $T(S_\rho(U))=t$,
does not depend on $|U|$.
\end{definition}

\subsection{Composable Cardinality Sketches} \label{sketches:sec}
The underlying technique in all small space cardinality sketches is to use random hashmaps $h$ that assign ``priorities'' to keys in $\mathcal{U}$.\footnote{\citep{Vinod:ESA2022} proposed a method for streaming cardinality estimates that does not require hashmaps but the sketch is not composable.} The sketch of a set is specified  by the priorities of a small set  of keys with the lowest priorities. 
This information is related to the cardinality as smaller lowest priorities in a subset $U$ are indicative of  larger cardinality. Therefore cardinality estimates can be recovered from the sketch. We describe several common designs.
MinHash sketches (see surveys~\citep{MinHash:Enc2008,Cohen:PODS2023}) are suitable for insertions only (set unions and insertions of new elements) and are also suitable for sketch-based sampling. Domain sampling has priorities that are  discretized sampling rates and has the advantage that the sketch can be represented as random linear maps (specified by $\rho$) of the data vector and therefore have support for deletions (negative entries in sketched  vectors)~\citep{distinct_deletions_Ganguly:2007}. 

\paragraph{MinHash sketches}

Types of MinHash sketches
\begin{itemize}
\item 
\emph{$k$-mins}~\citep{FlajoletMartin85,ECohen6f} $k$ random hash functions $h_1,\ldots,h_k$ that map each key $x\in \mathcal{U}$ to i.i.d samples from the domain of the hash function. The sketch $S_\rho(U)$ of a set $U$ is the list $(\min_{x\in U} h_i(x))_{i\in[k]}$ of minimum values of each hash function over the keys in $U$.  The sketch distribution for a subset $U$ is $\Exp[|U|]^k$, a set of $k$ i.i.d.\ exponentially distributed random variables with parameter $|U|$.  The sum $T(S) := \|S \|_1$ is a \emph{sufficient statistics} for estimating the parameter $|U|$. An unbiased cardinality estimator is $(k-1)/T(S)$.
\item 
\emph{Bottom-$k$}~\citep{Rosen1997a,ECohen6f,BJKST:random02} One random hash function $h$ that maps $x\in \mathcal{U}$ to i.i.d samples from a distribution. The sketch 
$\{h(x) \mid x\in U\}_{(1:k)}$
stores the $k$ smallest hash values of keys $x\in U$. The $k$th smallest value $T(S) := \{h(x) \mid x\in U\}_{(k)}$ is a sufficient statistics for estimating $|U|$. When the distribution is $U[0,1]$, the unbiased cardinality estimate is $(k-1)/T(S)$.
\item 
\emph{$k$-partition} \citep{flajolet2007hyperloglog}.  One hash $P:\mathcal{U}\to [k]$ randomly partition keys to $k$ parts.  One hash function $h:\mathcal{U}$ maps keys to i.i.d $\Exp[1]$. The sketch includes the minimum in each part $(\min_{x\in U \mid P(x)=i} h(x))_{i\in[k]}$.
\end{itemize}
Note that the choice of (continuous) distribution does not affect the information content in the sketch.
Variations of these sketches store rounded/truncated numbers (HLL~\citep{flajolet2007hyperloglog} stores a maximum negated exponent). When studying vulnerabilities of query response algorithms, the result is stronger when the full precision representation is available to them.

The cardinality estimates obtained with these sketches have NRMSE error \eqref{nrmse:eq} of $1/\sqrt{k}$.  
\begin{definition} [bias of the sketch] \label{sketchbias:def}
We say that the sketch $S_\rho(U)$ of a set $U$ is \emph{biased up} by a factor of $1/\alpha$ when
$T(S_\rho(U)) \leq \alpha k/|U|$ and we say it is \emph{biased down} by a factor of $\alpha$ when
$T(S_\rho(U)) \geq (1/\alpha) k/|U|$. 
\end{definition}
For our purposes, $\alpha\leq 1/2$ would places the sketch at the $\delta = e^{-\Omega(k)}$ tail of the distribution under sampling of $\rho$ and we say that $U$ is adversarial for $\rho$.

\paragraph{Domain sampling}
These cardinality sketches can be expressed as discretized bottom-$k$ sketches. Therefore vulnerabilities of bottom-$k$ sketches also apply with domain sampling sketches. 
The input is viewed as a vector of dimension $|\mathcal{U}|$ where the set $U$ corresponds to its nonzero entries. The cardinality $|U|$ is thus the sparsity (number of nonzero entries). The sketch map $S_\rho$ is a dimensionality reduction via a random linear map (specified by $\rho$). 

We sample the domain $\mathcal{U}=[n]$ with different rates $p=2^{-j}$. For each rate, we collect a count $c_j$ (capped by $k$) of the number of sampled keys from our set $X$.  This can be done by storing the first $k$ distinct keys we see or (approximately) by random hashing into a domain of size $k$ and considering how many cells were hit. 
A continuous version known as \emph{liquid legions}~\citep{49177}) is equivalent to a bottom-$k$ sketch: Each key is assigned a random i.i.d.\ priority (lowest sampling rate in which it is counted with domain sampling) and we seek  the sampling rate with which we have $k$ keys.

\paragraph{Specifying keys for the sketch}
Note that with all these sketch maps, the sketch of a set $U$ is specified by a small subset $U_0 \subset U$ of the ``lowest priority'' keys in $U$. With $k$-mins and $k$-partition sketches it is the keys $\arg\min_{x\in U}\{ h_i(x)\}$ for  $i\in [k]$. With bottom-$k$ sketches, it is  the keys with $k$ smallest values in $\{h_i(x)\}_{x\in U}$.  With  domain sampling, it is  the keys with  the highest  sampling rate.  Note that $|U_0| = O(k) \ll |U|$ but $S_\rho(U_0) = S_\rho(U)$.

\section{Attack on the ``standard'' estimators} \label{standardattack:sec}

The ``standard'' cardinality estimators optimally use the content in the sketch. They can be equivalently viewed as reporting a sufficient statistics.
We design a single-batch attack described in Algorithm~\ref{standardattack:algo}.  The algorithm fixes a ground set $N$ of keys. For $r$ queries, it samples a random subset $U\subset N$ where each $u\in N$ is included independently with probability $1/2$. It receives from the estimator the value of the sufficient statistics $T(S_\rho(U)) := 1/M(S_\rho(U))$ (we use the  inverse of the cardinality estimate). For each key $x\in N$ it computes a score $A[x]$ that is the average value of $T(S_\rho(U))$ over all subsets where $x\in U$.

\begin{algorithm2e}[t]\caption{\small{\texttt{Attack ``standard'' estimators}}}\label{standardattack:algo}
{\small
\DontPrintSemicolon
\KwIn{$\rho$,  $n$, $r$, $T$}
Fix a set $N$ of $n$ keys \tcp{selected randomly independently of $\rho$)}
\ForEach(\tcp*[f]{initialize}){key $x\in N$}{$t[x] \gets 0$ \\ $c[x] \gets 0$}
\ForEach{$i=1,\ldots,r$}
{
$U\gets$ include each $x\in N$ independently with prob $\frac{1}{2}$\; 
\ForEach(\tcp*[f]{score keys}){key $x\in U$}{$t[x] \gets t[x]+1$ \\ $c[x] \gets c[x] + T(S_\rho(U))$}
}
\Return{The keys in $N$ ordered by average score $A[x]=\frac{c[x]}{t[x]}$.}
}
\end{algorithm2e}

We show that for $\alpha>0$,  
an attack of size $O(r/\alpha^2)$ produces an adversarial set with sketch that is biased up by a factor $\alpha$ (see Definition~\ref{sketchbias:def}).
\begin{theorem} [Utility of Algorithm~\ref{standardattack:algo}] \label{stattackutility:thm}
     Consider Algorithm~\ref{standardattack:algo} with
    $k$-mins or bottom-$k$ sketches and $T(S)$ being the inverse of the cardinality estimate as specified in Section~\ref{sketches:sec}. For $\alpha > 0$, set the parameters $n=\Omega(\frac{1}{\alpha} k \log(kr))$ and $r= O\left(\frac{k}{\alpha^2} \right)$.
    Then with probability at least $0.99$, the sketch $S_\rho(U_\alpha)$, where $U_\alpha\subset N$ is the of the $\alpha n$ lowest $A[u]$ scores, is biased up by a factor of $\Omega(1/\alpha)$:
    \[
    \E\left[M(U_\alpha)\right] = \Theta(n)\ .
    \]
\end{theorem}
Our analysis extends to the case when the estimator reports $T(S_\rho(U))$ with relative error $O(1/\sqrt{k})$.  That is, as long as the estimates are sufficiently accurate (within the order of the accuracy guarantees of a size $k$ sketch), then $O(k)$ attack queries suffice.

\paragraph{Analysis Highlights}
The proof of Theorem~\ref{stattackutility:thm}  is presented in Appendix~\ref{kminsattack:sec}.
The high level idea is as follows. 
We establish that scores are correlated with the priorities of keys -- the keys with lowest priorities have in expectation lower scores. 
Therefore a prefix of the order will contain disproportionately more of them and overestimate the cardinality and a suffix will contain disproportionately fewer of them and underestimate the cardinality. 

We consider, for each key $x\in N$, the distributions of $T(S_\rho(U))$ conditioned on $x\in U$. We bound from above the variance of these distributions and bound from below the gap in the means of the distributions between the keys that have the ``lowest priority'' in $N$ and the bulk of other keys in $N$.  
We then apply Chebyshev's Inequality to bound the number of rounds that is needed so that enough of the low priority keys have lower average scores $A[u]$ than ``most'' other keys. A nuance we overcame in the analysis was to handle the dependence of the sketches of the different queries that are selected from the same ground set.

\section{Experimental Evaluation} \label{experiments:sec}

In this section, we empirically demonstrate the efficacy of our proposed attack (Algorithm~\ref{standardattack:algo}) against the HyperLogLog (HLL) sketch~\cite{DurandF:ESA03,flajolet2007hyperloglog} with the HLL++ estimator~\cite{hyperloglogpractice:EDBT2013}. This is 
the most widely utilized sketch for cardinality estimation in practice. 
Given an accuracy parameter $\epsilon$, the HLL sketch stores $k = 1.04 \epsilon^{-2}$ values that are the negated exponents of a $k$-partition MinHash sketch (described in Section~\ref{sketches:sec}). 

The HLL++ estimator is a hybrid that was introduced in order to improve accuracy on low cardinality queries. When the sketch representation is sparse (fewer parts are populated), which is the case with cardinality lower than the sketch size, HLL++ uses the sketch as a hash table and estimates cardinality based on the number of populated parts. This yields essentially precise values. When all parts are populated, HLL++ uses an estimator based on the MinHash property. %Note that higher cardinality is the interesting regime for sublinear sketches. 
We set the size of our ground set $n\gg k$ to be in this relevant regime. 

We conduct two primary experiments: (i) For a fixed sketch size, we analyze the efficacy of the attack with a varying number of queries. (ii) For different sketch sizes, we evaluate the effectiveness of the attack with the number of queries linearly proportional to the sketch size. In the following section, we will first provide a detailed explanation of the ingredients required for our experimental setup. Subsequently, we will present the results of each experiment.

\paragraph{Experiment setup.} To generate the data, we ensure that a ground set with a size of at least $10\cdot k$ is produced for a given sketch size $k$. The size of the ground set must be at least linearly larger than the sketch size to prevent the sketch from memorizing the entire dataset. Given the desired size of the ground set, we generate random strings using the English alphabet of a fixed length, where the length is appropriately chosen so that we can generate the desired size set with different strings. 

We utilize the open-source implementation of HLL++ algorithm in \href{https://github.com/svpcom/hyperloglog}{github}. In this implementation, the sketch is fixed by giving the error rate $\epsilon\in (0,1)$ and the sketch size $k$ for error rate $\epsilon$ is $\lceil 1.04/\epsilon^2 \rceil$ (consistent with \cite{flajolet2007hyperloglog}). 

\subsection{Efficacy with a varying number of queries}
In this experiment we examine the impact of introducing a variable quantity of queries. The attack is executed with the same ground set for eight distinct query counts, where each count is a power of 4. At the conclusion, the algorithm generates scores and returns keys sorted in ascending order according to their scores. Keys with high score correspond to low-priority keys which are expected to appear when the estimate is biased up. By including these keys in the adversarial set, we basically can trick the estimator to think that they are seeing a sketch of a large set. Similarly we can construct adversarial input sets by including keys with low scores and trick the estimator to think they are seeing a sketch of a small set.  

We present two sets of plots corresponding to how the estimator overestimates or underestimates as keys are incrementally added to the adversarial input in the increasing or decreasing order of their scores. We consider two different error rates, $\epsilon=0.1$ with corresponding sketch size $k=104$ and
$\epsilon=0.05$, with corresponding sketch size $k=416$. We use the same ground set comprising of $5000$ keys for both sets of experiments. It is worth noting that the plot with one query, which oscillates around the line $y = x$ (denoted by a dashed line), is close to a non-adversarial setting and we can see that the estimates are within the desired specified error of $\epsilon$. 

Figure~\ref{fig:HLLppattack05} reports cardinality estimates when keys are added incrementally in increasing order of their scores. We can see that as we increase the number of queries, the gap between estimated value and the $y=x$ line (actual value) widens. This gap indicate the overestimation error. Our algorithm is able to construct more effective adversarial input with a larger number of queries. However the gain in effectiveness becomes marginal at some point. For example, for $k = 104$, we already see good degree of error in estimation with $4096$ queries. 
\begin{figure}[h]
  \centering
  \includegraphics[width=0.4\textwidth]{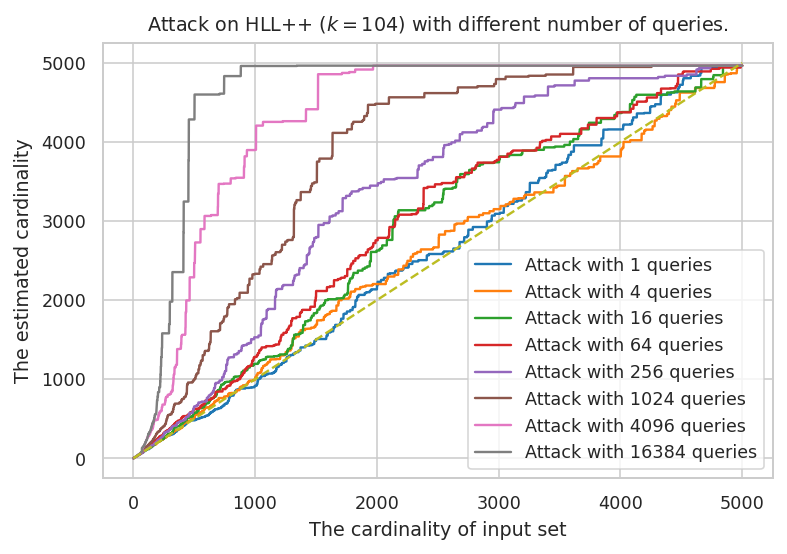}
  \includegraphics[width=0.4\textwidth]{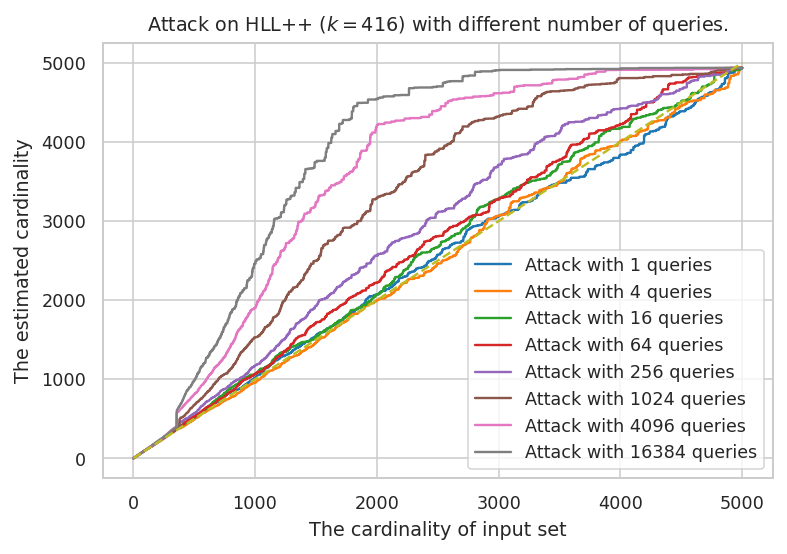}
  \caption{Attack on the HLL++ sketch and estimator, for varying number of queries. Cardinality estimates for the prefix of keys with  lowest average score after $r=4^i$ queries.}
  \label{fig:HLLppattack05}
\end{figure}
Figure~\ref{fig:HLLppattackdiffrounds} reports results when keys are added incrementally in \emph{decreasing} order of their scores. The gap here corresponds to an \emph{underestimation} error.  We can see that the attacks are more effective with more queries.
\begin{figure}[h]
  \centering
  \includegraphics[width=0.4\textwidth]{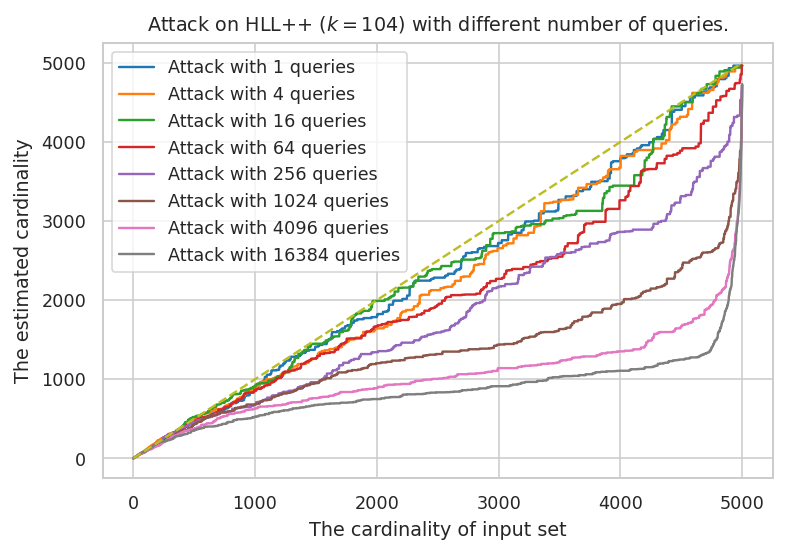}
  \includegraphics[width=0.4\textwidth]{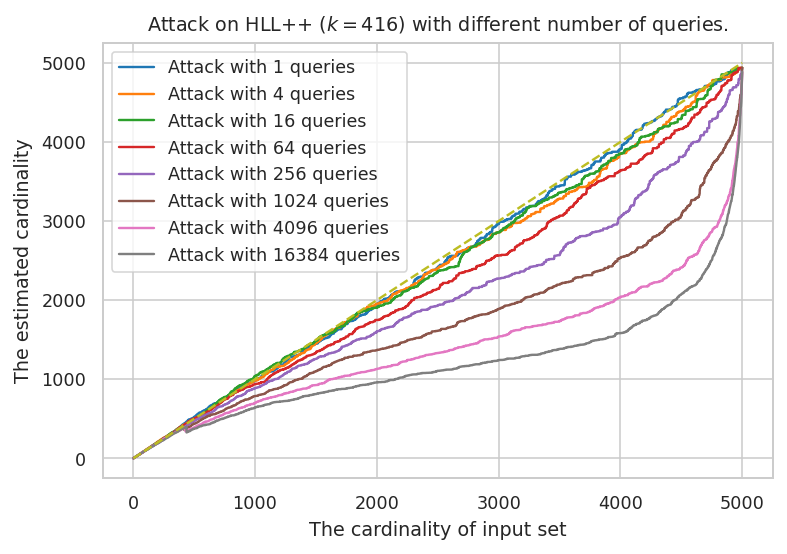}
  \caption{Attack on the HLL++ sketch and estimator, for varying number of queries. Cardinality estimates for the prefix of keys with largest average score after $r=4^i$ queries.}
  \label{fig:HLLppattackdiffrounds}
\end{figure}
\subsection{Efficacy of the attack with a varying sketch sizes} 
In this section, our focus is on examining HLL++ with different sketch sizes, namely we consider six different error rates corresponding to sketch sizes $k= 2^i$ for $i$ ranging from $6$ to $11$. For each sketch size $k$, we generate a ground set of size $n = 10*10^{\lceil log_{10}(k) \rceil}$ to ensure that the ground set is larger than sketch size and the MinHash component of the HLL++ estimator is used.  In Figure \ref{fig:HLLppattack2M}, we report the ratio of estimated size to actual size of the set for all subsets constructed as a prefix of the order on keys, sorted by increasing average scores $A[x]$ for a fixed number of queries set to $4k$.  In this cardinality regime, HLL++ is nearly unbiased and we expect a ratio that is close to $1$ when the queries are not adversarial. However by running attacks with enough number of queries (linear in the size of sketch), we are able to identify keys with low-priority and then trick the estimator to give an estimate for a set much higher than the actual size. 

\begin{figure}[h]
  \centering
  \includegraphics[width=0.4\textwidth]{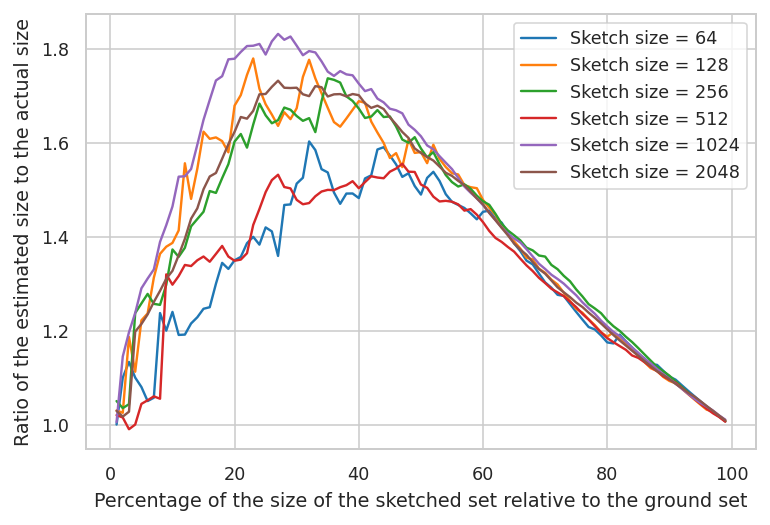}
  \caption{Attack on HLL++ for varying sketch sizes while utilizing queries of size 4 times the sketch size.}
  \label{fig:HLLppattack2M}
\end{figure}

\section{Attack Setup Against Strategic Estimators} \label{setupattack:sec}

We design attacks that apply generally against any query response (QR) algorithm. The attacks are effective even when the specifics of the attack and the full internal state of the attack algorithm are shared with the QR algorithm, including the per-step distribution from which the attacker selects each query. Moreover, we can even assume that the QR algorithm is provided with an enhanced sketch that includes the identities of the low priority keys that determined the sketch and that after the QR algorithm responds to a query, the full query set $U$ is shared with it. The only requirement from the QR algorithm is that it selects \emph{correct} response maps (with respect to the query distribution).  Note that such a powerful QR algorithm precludes attacks that use queries of fixed cardinality, since the QR algorithm can simply return that cardinality value without even considering the actual input sketch.\footnote{Our
attack in Algorithm~\ref{standardattack:algo} is also precluded, since a fixed response of $n/2$ satisfies the requirements (the cardinality is $\Binom(n,p)$ and for $n\gg k$, all queries have size close to $n/2$).}

Moreover, the task of the QR algorithm is the 
following problem that is more specialized than cardinality estimation:
\begin{problem} [Soft Threshold $A$]\label{softthresh:problem} Return $0$ when $|U| \leq A$ and $1$ when $|U| \geq 2A$. 
\end{problem}
\begin{remark} \label{softthreshold:rem}
   Soft Threshold can be solved via a cardinality estimate with a multiplicative error of at most $\sqrt{2}$ by reporting $1$ when the estimate is larger than $\sqrt{2}A$ and $0$ otherwise. When estimates are computed from cardinality sketches with randomness $\rho$ that does not depend on the queries, a sketch of size $k=\Theta(\log(1/\delta))$ is necessary and suffices for providing correct responses with probability $\geq 1-\delta$.
\end{remark}

\paragraph{Attack Framework}
We describe the attack framework. We specify attacks in Sections~\ref{singlebatch:sec} and~\ref{generalmain:sec}.
We model the interaction as a process between three parties: the \emph{Attacker}, the \emph{QR algorithm}, and \emph{System}.
The attacker fixes a ground set $N$ from which it samples query sets. The product of the attack is a subset $M\subset N$ which we refer to as a \emph{mask}.  The aim is for the mask to have size that is much smaller than our query subset sizes and the property that for uniformly sampled $U\subset N$ ($|U|\gg |M|$) the information in the sketch $S_\rho(M\cup U)$ is insufficient to estimate $|U|$. The attack proceeds in steps described in Algorithm~\ref{attackinteraction:algo}:
\begin{algorithm2e}\caption{Attack Interaction Step}\label{attackinteraction:algo}
\begin{small}
\begin{trivlist}
\item $\bullet$
Attacker
specifies a distribution $\mathcal{D}$ over the support $2^N$ and sends it to QR. It then selects a query $U\sim \mathcal{D}$ and sends it to \emph{System}.
\item $\bullet$
QR selects a \emph{correct}
map with $\delta=O(1/\sqrt{k})$ (as in Definition~\ref{correctmap:def}) of sketches to probabilities
$S \mapsto \pi(S)\in [0,1]$.  The selection may depend on the prior interaction transcript and on $\mathcal{D}$. If there is no correct map QR reports \textbf{failure} and halts.
\item $\bullet$ \emph{System} computes the sketch $S_\rho(U)$ and sends it to QR.
\item $\bullet$
QR sends
$Z\sim \Bern[\pi(S_\rho(U))]$ to Attacker. Attacker shares $U$ and its internal state with QR.
\end{trivlist}
\end{small}
\end{algorithm2e}
\begin{definition} [Correct Map] \label{correctmap:def}
   We say that the map $\pi$ is \emph{correct} for
$A$ and $\delta$ and query distribution $\mathcal{D}$
if \ignore{it returns a correct response 
to a soft threshold problem with $A$ with probability at least $1-\delta$.  That is,
\begin{align*} \label{correctmap:eq}
\E_{U\sim\mathcal{D}}[ \mathbf{1}\{|U|<A\}\cdot (\pi(S_\rho(U))) + \\ \mathbf{1}\{|U|>2A\}\cdot (1-\pi(S_\rho(U)))]  &\leq \delta 
\end{align*} 
\edith{Version 2 (stronger, simpler proof later)}
If} for any cardinality value, over the query distribution for this value, it returns a correct response to a soft threshold problem with $A$ with probability at least $1-\delta$.  That is,
\begin{align*} 
\text{for $c<A$, } \E_{U\sim\mathcal{D} \mid |U|=c}(\pi(S_\rho(U))) &\leq \delta \\ 
\text{for $c>2A$, } \E_{U\sim\mathcal{D} \mid |U|=c}(\pi(S_\rho(U))) &\geq 1-\delta\ .
\end{align*} 
\end{definition}

\begin{remark} [Many correct maps] \label{manycorrect:rem}
There can be multiple correct maps and QR may choose any one at any step. 
Since the output when $|U|\in [A,2A]$ is not specified, the probability $\E_{U\sim\mathcal{D}}\pi(S_\rho(U))$ of reporting  $Z=1$ may vary by 
$\approx \Pr_{U\sim\mathcal{D}}[|U|\in [A,2A]] +\delta$ between 
correct maps.
\end{remark}

Recall that our attack on the standard estimators (Algorithm~\ref{standardattack:algo}) issued a single batch of queries (all drawn from the same pre-specified distribution $\mathcal{D}_0$). We show that multiple batches are necessary to attack general QR algorithms:
\begin{lemma} [Multiple batches are necessary] \label{adaptivenecessary:lemma}
    Any attack of polynomial size in $k$ on a soft threshold estimator must use multiple batches.
\end{lemma}
\begin{proof}
When there is a single batch of $r$ queries,  we can apply the standard estimator while accessing only a ``component'' of the sketch that is of size $k'=O(\log (r/\delta))$ and obtain correct responses on all queries. This component is the same $k'$ hash functions with $k$-mins sketches, only $k'$ parts in a $k$-partition sketch, or the bottom-$k'$ values in a bottom-$k$ sketch. Since the query response only leaks information on this component, the attacker is only able to compromise that component in this single-batch attack. Therefore, an exponential number of queries in $k$ is needed in order to construct an adversarial input in a single batch. 
\end{proof}

\section{Single-batch attack on symmetric QR} \label{singlebatch:sec}

Algorithm~\ref{onebatchgen:algo} specifies a single-batch attack. We establish that the attack succeeds when we set the size $r=\tilde{O}(k^2)$ and QR is constrained to be \emph{symmetric} (see Definition~\ref{symmetric:def}). In essence symmetry means that the QR algorithm does not make a significant distinction between components of the sketch. Symmetry excludes strategies that distinguish between components of the sketch as in the proof of  Lemma~\ref{adaptivenecessary:lemma} but still allows for flexibility including randomly selecting components of the sketch. 
In Section~\ref{generalmain:sec} we extend this attack  to an adaptive attack that works against \emph{any} QR algorithm.

\begin{algorithm2e}[t]\caption{\small{\texttt{Single Batch Attacker}} \label{onebatchgen:algo}}
{\small
\DontPrintSemicolon
\KwIn{$\rho$,  $n$, $r$}
Select a set $N$ of $n$ keys\tcp*{Randomly from $\mathcal{U}$}
$A \gets n/16$\; 
\lForEach(\tcp*[f]{initialize}){key $x\in N$}{$C[x] \gets 0$}
\For{$i=1,\ldots,r$}
{
Sample $q$ as specified in Algorithm~\ref{samplerate:algo}\tcp*{using $A,n$}
$U\gets $ includes each $u\in N$ independently with prob $q$\;
\textbf{send} $U$ \tcp*{$\to$ System}
\textbf{receive} $Z$ \tcp{$\gets$ Symmetric Query Response}
\lForEach(\tcp*[f]{score}){key $x\in U$}{$C[x] \gets C[x]+Z$}
}
$\overline{C}\gets \textrm{median}\{C[N]\}$\tcp*{Compute median score}
\Return{$M \gets \{x\in N \mid C[x] > \overline{C} + \tilde{\Omega}(\frac{r}{k})\}$}\tcp*{Mask}
}
\end{algorithm2e}

The attacker initializes the scores $C[x]\gets  0$ of all keys $x\in N$ in the ground set. Each query is formed by 
sampling a 
rate $q$ (as described in Algorithm~\ref{samplerate:algo}) and selecting a random subset $U \subset N$ so that each key in $N$ is included independently with probability $q$. The attacker receives $Z$ and increments by $Z$ the score $C[x]$ of all keys $x\in U$. The final product is the set $M$ of keys with scores that are higher by $\tilde{\Omega}(r/k)$ than the median score.

\begin{theorem} [Utility of Algorithm~\ref{onebatchgen:algo} with symmetric maps] \label{onebatchutility:thm}
     For $\alpha > 0$, set $n=\Omega(\frac{1}{\alpha} k \log(kr))$ and $r= \tilde{\Omega}\left(\frac{k^2}{\alpha^2} \right)$. 
     Then
     $\Pr[(S_\rho(M) = S_\rho(N)) \land  (|M|<\alpha n)] \geq 0.99$. 
\end{theorem}

\subsection{Proof Overview} \label{onebatchhighlights:sec}
See Appendix~\ref{genattackproofs:sec} for details.
We work with the \emph{rank-domain} representations of the sketches with respect to the ground set $N$.  This representation simplifies our analysis as it only depends on the rank order of keys by their hash values and by that factors out the hash values.
The rank-domain sketches $S^R(U)$ have the form
$(Y_1,\ldots,Y_k)$ where $Y_i$ are positive integers in $[N]$.
The sketch distribution over the sampling of $U$ for fixed $q$ is that of $k$ independent $\Geom[q]$ random variables $Y_i$.  
The sum $T = \sum_{i=1}^k Y_i$ is
a sufficient statistics for $q$ from the sketch.

\begin{definition} [symmetric map]\label{symmetric:def}
    A map $\pi$ is \emph{symmetric} if it uses the rank-domain sketch as an \emph{unordered} set and (ii)~is monotone in that if a sketch $S_1 \leq S_2$ coordinate-wise then then $\pi(S_1) \geq \pi(S_2)$.
\end{definition}

We denote by 
$N_0^*$ the set of $k$ lowest-rank (lowest priority) keys. It includes the bottom-$k$ keys with bottom-$k$ sketches, and the minimum hash key with respect to each of the $k$ hashmaps with $k$-mins and $k$-partition sketches. Note that $S_\rho(N) = S_\rho(N_0^*)$. We denote by $N'\subset N$ a set of keys that
are \emph{transparent} -- very unlikely to influence the sketch if included in the attack subsets of Algorithm~\ref{onebatchgen:algo}. 
We show that a key in $N_0^*$ obtains in expectation a higher score than a key in $N'$ and use that to establish the utility claim:
\begin{lemma}  [separation with symmetric maps] \label{symmetricsep:lemma}
Let $\pi$ be correct and symmetric  (Definition~\ref{symmetric:def}). Then for any $m\in N^*_0$ and $u\in N'$,
{\small
\[
\E_{U}[\pi(S_\rho(U)\cdot \mathbf{1}(m\in U)] - \E_{U}[\pi(S_\rho(U)\cdot \mathbf{1}(u\in U)] = \tilde{\Omega}(1/k)
\]}
\end{lemma}
\begin{proof}
  See Appendix~\ref{symmetriccase:sec}. The gap only holds on average for general correct maps (Lemma~\ref{scoregap:lemma}) but holds per-key when specialized to symmetric maps.
\end{proof}
% We are now ready to conclude the utility claim:
\begin{proofof}{Theorem~\ref{onebatchutility:thm}}
The distribution of the score $C[x]$ of all transparent keys $u\in N'$ is identical and is the sum of $r$ independent Poisson random variables 
$\sum_{i=1}^r Z_i\cdot \Bern[q_i]$.  

From Lemma~\ref{symmetricsep:lemma} 
$\E[C[m]] - \E[C[u]] = \tilde{\Omega}(r/k)$. 
For $x\in N^*_0$, the gap random variables of different steps may be dependent, but the lower bound on the expected gap in Lemma~\ref{symmetricsep:lemma} holds even conditioned on transcript. Additionally, the expected gap in each step is bounded in $[-1,1]$. We 
% https://people.eecs.berkeley.edu/~sinclair/cs271/n13.pdf
can apply Chernoff bounds~\citep{Chernoff52} to bound the probability that a sum deviates by more than $\lambda$ from its expectation 
\begin{equation} \label{Chernoff:eq}
\Pr[|C[x]-\E[C[x]]| \geq \lambda ] \leq 2e^{-2\lambda^2/r}\ . 
\end{equation}
Setting $\lambda = cr/k$ separates a key in $N_0^*$ from a key in $N'$ with probability $1- 2e^{-2c^2 r/ k^2}$. 
Choosing $r=O(k^2\log |N|)$ we get that the order separates out with high probability all the keys $N^*_0$ from all the keys $N'$. 
Note that there are only 
$\tilde{\Omega}(k)$ non transparent keys $N_0 :=N\setminus N'$. Therefore with high probability $N^*_0\subset M\subset N_0$ and (we can fix constants so that) $|M|=\tilde{O}(k)\leq \alpha n$. Since $N^*_0\subset M$, $S_\rho(M) = S_\rho(N)$. 
\end{proofof}

\section{Adaptive Attack on General QR} \label{generalmain:sec}

\begin{algorithm2e}[t]\caption{\small{\texttt{Adaptive Attacker}} \label{adaptivegen:algo}}
{\small
\DontPrintSemicolon
\KwIn{$\rho$,  $n$, $r$}
Select a set $N$ of $n$ keys\tcp*{Randomly from $\mathcal{U}$}
$A \gets n/16$; $M\gets \emptyset$\;
\lForEach(\tcp*[f]{initialize}){key $x\in N$}{$C[x] \gets 0$}
\For{$i=1,\ldots,r$}
{
Sample $U\sim \mathcal{D}_0$ as in Algorithm~\ref{onebatchgen:algo}\;
\textbf{send} $M\cup U$ to system\;
\textbf{receive} $Z$ from QR\;
\lIf{\textbf{failure}}{\textbf{exit}}
\ForEach(\tcp*[f]{score keys}){key $x\in U$}{$C[x] \gets C[x]+Z$\;
\If(\tcp*[f]{test if score is high}){$C[x] \geq  \textrm{median}(C[N\setminus M]) + \sqrt{i \log(200 n r)/2}$}{$M\gets M\cup\{x\}$}}
\textbf{send} $M,C,U$ to QR \tcp*{share internal state}
}
\Return{$M$}\tcp*{Mask}
}
\end{algorithm2e}

An attack on general QR algorithms is given in Algorithm~\ref{adaptivegen:algo}. 
The attacker  maintains an initially empty set $M\subset N$ of keys which we refer to as \emph{mask}. The query sets have the form $M\cup U$, where $U\sim \mathcal{D}_0$ is sampled and scored as in Algorithm~\ref{onebatchgen:algo}.
A key is added to $M$ when its score separates out from the median score. 
We establish the following:
\begin{theorem} [Utility of Algorithm~\ref{adaptivegen:algo}] \label{adaptivegenutility:thm}
     For $\alpha > 0$, set $n=\Omega(\frac{1}{\alpha} k \log(kr))$ and $r= \tilde{\Omega}\left(\frac{k^2}{\alpha^2} \right)$. 
     Then with probability at least $0.99$, $|M|<\alpha n$ and there is no correct map for the query distribution $M\cup U$ where $U\sim\mathcal{D}_0$ is as in Algorithm~\ref{onebatchgen:algo}.
\end{theorem}
We overview the proof with details deferred to Appendix~\ref{genattackproofs:sec}. 
The condition for adding a key to $M$ is such that 
with probability at least $0.99$, only $N_0$ keys are placed in $M$, so $|M| \leq \alpha n$
(Claim~\ref{notransparentinM:claim}).
If the QR algorithm fails, there is no correct map for the distribution $M\cup\mathcal{D}_0$.\footnote{A situation of no correct maps can be identified by Attacker, by tracking the error rate of QR, even if not declared by QR.} It remains to consider the case where the attack is not halted.

Since the mask $M$ is shared with QR ``for free,'' QR only needs to estimate $|U|$ (or $q$). But the sketch of $M\cup U$ partially masks the sketch of $U$. The set of non-transparent keys $N'_0\subset N_0$ decreases as $M$ increases. Additionally, the \emph{effective} sketch size $k' \leq k$ is lower (that is, QR only obtains $k'$ i.i.d\ $\Geom[q]$ random variables). Recall that when $k' < \log(k)/2$, there is no correct map. 
% As $M$ is augmented, additional keys from $N_0$ become transparent. The remaining non-transparent keys $N'_0 \subset N_0$ 
% $|N'_0|= \tilde{O}(k')$.

With general correct maps, we can only establish a 
weaker 
\emph{average} version of the score gap over $N'_0$ keys.
This allows some $N'_0$ keys to remain indistinguishable by score from transparent keys. But what works in Attacker's favour is that in this case the score of other $N_0$ keys must increase faster. 
Let $p(\pi,M,x)$ be the probability that key $x$ is scored with map $\pi$ and mask $M$.
The probability is the same for all transparent keys  $x\not\in N'_0$ and we denote it by $p'(\pi,M)$.
We establish (see Lemma~\ref{scoregapmask:lemma}) that for a  correct map $\pi$ for $M\cup\mathcal{D}_0$ it holds that
\[
\sum_{x\in N'_0} \left( p(\pi,M,x) - p'(\pi,M) \right) = \tilde{\Omega}(1)\ .
\]
Therefore, in $r=\tilde{O}(k^2)$ steps, the combined score advantage of $N_0$ keys is (concentrated well around) $\tilde{O}(k^2)$.
But crucially, any one key can not get too much advantage: once $C(x) - \overline{C} = \tilde{\Omega}(k)$ (where $\overline{C}$ is the median score), then key $x$ is placed in the mask $M$, exits $N'_0$, and stops getting scored. Therefore if QR does not fail,  $\tilde{\Omega}(r/k) > |N_0|$ keys are eventually placed in $M$, which must include all $N_0$ keys.

\section{Conclusion}
We demonstrated the inherent vulnerability of the known composable cardinality sketches to adaptive inputs. We designed attacks that use a number of queries that asymptotically match the upper bounds: A linear number of queries with the ``standard'' estimator and a quadratic number of queries with \emph{any} estimator applied to the sketch. Empirically, our attacks are simple and effective with small constants. An interesting direction for further study is to show that this vulnerability applies with any composable sketch structure.  On the positive side, we 
suspect that restricting the maximum number of queries that any one key can participate in to sublinear (with standard estimators) or subquadratic (with general estimators) would enhance robustness.

\newpage
\section*{Acknowledgements}
The authors are grateful to Jelani Nelson and Uri Stemmer for discussions. Edith Cohen is partially supported by Israel Science Foundation (grant no. 1156/23).

% \section*{Impact Statement}
%This paper presents work whose goal is to advance the field of Machine Learning. There are many potential societal consequences of our work, none which we feel must be specifically highlighted here.

% \bibliographystyle{icml2024}
\bibliographystyle{plainnat}
\bibliography{main,references,robustHH}

\begin{thebibliography}{63}
\providecommand{\natexlab}[1]{#1}
\providecommand{\url}[1]{\texttt{#1}}
\expandafter\ifx\csname urlstyle\endcsname\relax
  \providecommand{\doi}[1]{doi: #1}\else
  \providecommand{\doi}{doi: \begingroup \urlstyle{rm}\Url}\fi

\bibitem[Ahn et~al.(2012)Ahn, Guha, and McGregor]{AhnGM:SODA2012}
Kook~Jin Ahn, Sudipto Guha, and Andrew McGregor.
\newblock Analyzing graph structure via linear measurements.
\newblock In \emph{Proceedings of the 2012 Annual ACM-SIAM Symposium on
  Discrete Algorithms (SODA)}, pages 459--467, 2012.
\newblock \doi{10.1137/1.9781611973099.40}.
\newblock URL \url{https://epubs.siam.org/doi/abs/10.1137/1.9781611973099.40}.

\bibitem[{Apache Software Foundation}(Accessed: 2024)]{datasketches}
{Apache Software Foundation}.
\newblock {DataSketches}, Accessed: 2024.
\newblock URL \url{https://datasketches.apache.org}.
\newblock Apache Software Foundation Documentation.

\bibitem[Athalye et~al.(2018)Athalye, Engstrom, Ilyas, and
  Kwok]{athalye2018synthesizing}
Anish Athalye, Logan Engstrom, Andrew Ilyas, and Kevin Kwok.
\newblock Synthesizing robust adversarial examples.
\newblock In \emph{International conference on machine learning}, pages
  284--293. PMLR, 2018.

\bibitem[Attias et~al.(2021)Attias, Cohen, Shechner, and Stemmer]{AttiasCSS21}
Idan Attias, Edith Cohen, Moshe Shechner, and Uri Stemmer.
\newblock A framework for adversarial streaming via differential privacy and
  difference estimators.
\newblock \emph{CoRR}, abs/2107.14527, 2021.

\bibitem[Bar-Yossef et~al.(2002)Bar-Yossef, Jayram, Kumar, Sivakumar, and
  Trevisan]{BJKST:random02}
Z.~Bar-Yossef, T.~S. Jayram, R.~Kumar, D.~Sivakumar, and L.~Trevisan.
\newblock Counting distinct elements in a data stream.
\newblock In \emph{RANDOM}. ACM, 2002.

\bibitem[Bassily et~al.(2021)Bassily, Nissim, Smith, Steinke, Stemmer, and
  Ullman]{BassilyNSSSU:sicomp2021}
Raef Bassily, Kobbi Nissim, Adam~D. Smith, Thomas Steinke, Uri Stemmer, and
  Jonathan~R. Ullman.
\newblock Algorithmic stability for adaptive data analysis.
\newblock \emph{{SIAM} J. Comput.}, 50\penalty0 (3), 2021.
\newblock \doi{10.1137/16M1103646}.
\newblock URL \url{https://doi.org/10.1137/16M1103646}.

\bibitem[Beimel et~al.(2022)Beimel, Kaplan, Mansour, Nissim, Saranurak, and
  Stemmer]{BKMNSS22}
Amos Beimel, Haim Kaplan, Yishay Mansour, Kobbi Nissim, Thatchaphol Saranurak,
  and Uri Stemmer.
\newblock Dynamic algorithms against an adaptive adversary: generic
  constructions and lower bounds.
\newblock page 1671–1684, 2022.
\newblock \doi{10.1145/3519935.3520064}.
\newblock URL \url{https://doi.org/10.1145/3519935.3520064}.

\bibitem[Ben{-}Eliezer et~al.(2021{\natexlab{a}})Ben{-}Eliezer, Eden, and
  Onak]{BEO21}
Omri Ben{-}Eliezer, Talya Eden, and Krzysztof Onak.
\newblock Adversarially robust streaming via dense-sparse trade-offs.
\newblock \emph{CoRR}, abs/2109.03785, 2021{\natexlab{a}}.

\bibitem[Ben{-}Eliezer et~al.(2021{\natexlab{b}})Ben{-}Eliezer, Jayaram,
  Woodruff, and Yogev]{BenEliezerJWY21}
Omri Ben{-}Eliezer, Rajesh Jayaram, David~P. Woodruff, and Eylon Yogev.
\newblock A framework for adversarially robust streaming algorithms.
\newblock \emph{{SIGMOD} Rec.}, 50\penalty0 (1):\penalty0 6--13,
  2021{\natexlab{b}}.

\bibitem[Boneh and Shaw(1998)]{BonehShaw_fingerprinting:1998}
Dan Boneh and James Shaw.
\newblock Collusion-secure fingerprinting for digital data.
\newblock \emph{{IEEE} Trans. Inf. Theory}, 44\penalty0 (5):\penalty0
  1897--1905, 1998.
\newblock \doi{10.1109/18.705568}.
\newblock URL \url{https://doi.org/10.1109/18.705568}.

\bibitem[Chakraborty et~al.(2022)Chakraborty, Vinodchandran¹, and
  Meel]{Vinod:ESA2022}
Sourav Chakraborty, N.~V. Vinodchandran¹, and Kuldeep~S. Meel.
\newblock Distinct elements in streams: An algorithm for the (text) book.
\newblock Schloss Dagstuhl – Leibniz-Zentrum für Informatik, 2022.
\newblock \doi{10.4230/LIPICS.ESA.2022.34}.
\newblock URL
  \url{https://drops.dagstuhl.de/entities/document/10.4230/LIPIcs.ESA.2022.34}.

\bibitem[Cherapanamjeri and Nelson(2020)]{DBLP:conf/nips/CherapanamjeriN20}
Yeshwanth Cherapanamjeri and Jelani Nelson.
\newblock On adaptive distance estimation.
\newblock In \emph{Advances in Neural Information Processing Systems 33: Annual
  Conference on Neural Information Processing Systems 2020, NeurIPS 2020,
  December 6-12, 2020, virtual}, 2020.

\bibitem[Chernoff(1952)]{Chernoff52}
H.~Chernoff.
\newblock A measure of the asymptotic efficiency for test of a hypothesis based
  on the sum of observations.
\newblock \emph{Annals of Math. Statistics}, 23:\penalty0 493--509, 1952.

\bibitem[Cohen(1997)]{ECohen6f}
E.~Cohen.
\newblock Size-estimation framework with applications to transitive closure and
  reachability.
\newblock \emph{Journal of Computer and System Sciences}, 55:\penalty0
  441--453, 1997.

\bibitem[Cohen(2015)]{ECohenADS:TKDE2015}
E.~Cohen.
\newblock All-distances sketches, revisited: {HIP} estimators for massive
  graphs analysis.
\newblock \emph{TKDE}, 2015.
\newblock URL \url{http://arxiv.org/abs/1306.3284}.

\bibitem[Cohen(2008)]{MinHash:Enc2008}
Edith Cohen.
\newblock \emph{Min-Hash Sketches}, pages 1--7.
\newblock Springer US, Boston, MA, 2008.
\newblock ISBN 978-3-642-27848-8.
\newblock \doi{10.1007/978-3-642-27848-8_573-1}.
\newblock URL \url{https://doi.org/10.1007/978-3-642-27848-8_573-1}.

\bibitem[Cohen(2018)]{CapSampling}
Edith Cohen.
\newblock Stream sampling framework and application for frequency cap
  statistics.
\newblock \emph{ACM Trans. Algorithms}, 14\penalty0 (4):\penalty0 52:1--52:40,
  2018.
\newblock ISSN 1549-6325.
\newblock \doi{10.1145/3234338}.

\bibitem[Cohen(2023)]{Cohen:PODS2023}
Edith Cohen.
\newblock Sampling big ideas in query optimization.
\newblock In Floris Geerts, Hung~Q. Ngo, and Stavros Sintos, editors,
  \emph{Proceedings of the 42nd {ACM} {SIGMOD-SIGACT-SIGAI} Symposium on
  Principles of Database Systems, {PODS} 2023, Seattle, WA, USA, June 18-23,
  2023}, pages 361--371. {ACM}, 2023.
\newblock \doi{10.1145/3584372.3589935}.
\newblock URL \url{https://doi.org/10.1145/3584372.3589935}.

\bibitem[Cohen and Geri(2019)]{CohenGeri:NeurIPS2019}
Edith Cohen and Ofir Geri.
\newblock Sampling sketches for concave sublinear functions of frequencies.
\newblock In \emph{NeurIPS}, 2019.

\bibitem[Cohen et~al.(2022{\natexlab{a}})Cohen, Lyu, Nelson, Sarl{\'{o}}s,
  Shechner, and Stemmer]{DBLP:conf/icml/CohenLNSSS22}
Edith Cohen, Xin Lyu, Jelani Nelson, Tam{\'{a}}s Sarl{\'{o}}s, Moshe Shechner,
  and Uri Stemmer.
\newblock On the robustness of countsketch to adaptive inputs.
\newblock In \emph{{ICML}}, volume 162 of \emph{Proceedings of Machine Learning
  Research}, pages 4112--4140. {PMLR}, 2022{\natexlab{a}}.

\bibitem[Cohen et~al.(2022{\natexlab{b}})Cohen, Nelson, Sarlós, and
  Stemmer]{TrickingHashingTrick:arxiv2022}
Edith Cohen, Jelani Nelson, Tamás Sarlós, and Uri Stemmer.
\newblock Tricking the hashing trick: A tight lower bound on the robustness of
  {C}ount{S}ketch to adaptive inputs.
\newblock \emph{arXiv:2207.00956}, 2022{\natexlab{b}}.
\newblock \doi{10.48550/ARXIV.2207.00956}.
\newblock URL \url{https://arxiv.org/abs/2207.00956}.

\bibitem[Desfontaines et~al.(2019)Desfontaines, Lochbihler, and
  Basin]{DLB:PET2019}
Damien Desfontaines, Andreas Lochbihler, and David~A. Basin.
\newblock Cardinality estimators do not preserve privacy.
\newblock In \emph{Privacy Enhancing Technologies Symposium}, volume~19, 2019.
\newblock \doi{https://doi.org/10.2478/popets-2019-0018}.
\newblock URL \url{http://arxiv.org/abs/1808.05879}.

\bibitem[Durand and Flajolet(2003)]{DurandF:ESA03}
M.~Durand and P.~Flajolet.
\newblock Loglog counting of large cardinalities (extended abstract).
\newblock In \emph{ESA}, 2003.

\bibitem[Dwork et~al.(2006)Dwork, McSherry, Nissim, and Smith]{DMNS06}
Cynthia Dwork, Frank McSherry, Kobbi Nissim, and Adam Smith.
\newblock Calibrating noise to sensitivity in private data analysis.
\newblock In \emph{TCC}, 2006.

\bibitem[Dwork et~al.(2015{\natexlab{a}})Dwork, Feldman, Hardt, Pitassi,
  Reingold, and Roth]{DworkFHPRR15}
Cynthia Dwork, Vitaly Feldman, Moritz Hardt, Toniann Pitassi, Omer Reingold,
  and Aaron~Leon Roth.
\newblock Preserving statistical validity in adaptive data analysis.
\newblock In \emph{{STOC}}, pages 117--126. {ACM}, 2015{\natexlab{a}}.

\bibitem[Dwork et~al.(2015{\natexlab{b}})Dwork, Feldman, Hardt, Pitassi,
  Reingold, and Roth]{DworkFHPRR:STOC2015}
Cynthia Dwork, Vitaly Feldman, Moritz Hardt, Toniann Pitassi, Omer Reingold,
  and Aaron~Leon Roth.
\newblock Preserving statistical validity in adaptive data analysis.
\newblock In \emph{Proceedings of the Forty-Seventh Annual ACM Symposium on
  Theory of Computing}, STOC '15, page 117–126, New York, NY, USA,
  2015{\natexlab{b}}. Association for Computing Machinery.
\newblock ISBN 9781450335362.
\newblock \doi{10.1145/2746539.2746580}.
\newblock URL \url{https://doi.org/10.1145/2746539.2746580}.

\bibitem[Flajolet and Martin(1985)]{FlajoletMartin85}
P.~Flajolet and G.~N. Martin.
\newblock Probabilistic counting algorithms for data base applications.
\newblock \emph{Journal of Computer and System Sciences}, 31:\penalty0
  182--209, 1985.

\bibitem[Flajolet et~al.(2007{\natexlab{a}})Flajolet, Fusy, Gandouet, and
  Meunier]{hyperloglog:2007}
P.~Flajolet, E.~Fusy, O.~Gandouet, and F.~Meunier.
\newblock Hyperloglog: The analysis of a near-optimal cardinality estimation
  algorithm.
\newblock In \emph{Analysis of Algorithms (AofA)}. DMTCS, 2007{\natexlab{a}}.

\bibitem[Flajolet et~al.(2007{\natexlab{b}})Flajolet, Fusy, Gandouet, and
  Meunier]{flajolet2007hyperloglog}
Philippe Flajolet, {\'E}ric Fusy, Olivier Gandouet, and Fr{\'e}d{\'e}ric
  Meunier.
\newblock Hyperloglog: the analysis of a near-optimal cardinality estimation
  algorithm.
\newblock \emph{Discrete mathematics \& theoretical computer science},
  \penalty0 (Proceedings), 2007{\natexlab{b}}.

\bibitem[Freedman(1983)]{Freedman:1983}
David~A. Freedman.
\newblock A note on screening regression equations.
\newblock \emph{The American Statistician}, 37\penalty0 (2):\penalty0 152--155,
  1983.
\newblock \doi{10.1080/00031305.1983.10482729}.
\newblock URL
  \url{https://www.tandfonline.com/doi/abs/10.1080/00031305.1983.10482729}.

\bibitem[Ganguly(2007)]{distinct_deletions_Ganguly:2007}
Sumit Ganguly.
\newblock Counting distinct items over update streams.
\newblock \emph{Theoretical Computer Science}, 378\penalty0 (3):\penalty0
  211--222, 2007.
\newblock ISSN 0304-3975.
\newblock \doi{https://doi.org/10.1016/j.tcs.2007.02.031}.
\newblock URL
  \url{https://www.sciencedirect.com/science/article/pii/S0304397507001223}.
\newblock Algorithms and Computation.

\bibitem[Gawrychowski et~al.(2020)Gawrychowski, Mozes, and
  Weimann]{gawrychowskiMW:ICALP2020}
Pawel Gawrychowski, Shay Mozes, and Oren Weimann.
\newblock Minimum cut in o(m log{\({^2}\)} n) time.
\newblock In \emph{{ICALP}}, volume 168 of \emph{LIPIcs}, pages 57:1--57:15.
  Schloss Dagstuhl - Leibniz-Zentrum f{\"{u}}r Informatik, 2020.

\bibitem[Goodfellow et~al.(2014)Goodfellow, Shlens, and
  Szegedy]{goodfellow2014explaining}
Ian~J Goodfellow, Jonathon Shlens, and Christian Szegedy.
\newblock Explaining and harnessing adversarial examples.
\newblock \emph{arXiv preprint arXiv:1412.6572}, 2014.

\bibitem[{Google Cloud}(Accessed: 2024)]{bigquerydocs}
{Google Cloud}.
\newblock {BigQuery Documentation: Approximate Aggregate Functions}, Accessed:
  2024.
\newblock URL
  \url{https://cloud.google.com/bigquery/docs/reference/standard-sql/approximate_aggregate_functions}.
\newblock Google Cloud Documentation.

\bibitem[Gutenberg and Wulff-Nilsen(2020)]{GutenbergPW:SODA2020}
Maximilian~Probst Gutenberg and Christian Wulff-Nilsen.
\newblock Decremental {SSSP} in weighted digraphs: Faster and against an
  adaptive adversary.
\newblock In \emph{Proceedings of the Thirty-First Annual ACM-SIAM Symposium on
  Discrete Algorithms}, SODA '20, page 2542–2561, USA, 2020. Society for
  Industrial and Applied Mathematics.

\bibitem[Hardt and Ullman(2014)]{HardtUllman:FOCS2014}
M.~Hardt and J.~Ullman.
\newblock Preventing false discovery in interactive data analysis is hard.
\newblock In \emph{2014 IEEE 55th Annual Symposium on Foundations of Computer
  Science (FOCS)}, pages 454--463. IEEE Computer Society, 2014.
\newblock \doi{10.1109/FOCS.2014.55}.
\newblock URL \url{https://doi.ieeecomputersociety.org/10.1109/FOCS.2014.55}.

\bibitem[Hardt and Woodruff(2013)]{HardtW:STOC2013}
Moritz Hardt and David~P. Woodruff.
\newblock How robust are linear sketches to adaptive inputs?
\newblock In \emph{Proceedings of the Forty-Fifth Annual ACM Symposium on
  Theory of Computing}, STOC '13, page 121–130, New York, NY, USA, 2013.
  Association for Computing Machinery.
\newblock ISBN 9781450320290.
\newblock \doi{10.1145/2488608.2488624}.
\newblock URL \url{https://doi.org/10.1145/2488608.2488624}.

\bibitem[Hassidim et~al.(2020)Hassidim, Kaplan, Mansour, Matias, and
  Stemmer]{HassidimKMMS20}
Avinatan Hassidim, Haim Kaplan, Yishay Mansour, Yossi Matias, and Uri Stemmer.
\newblock Adversarially robust streaming algorithms via differential privacy.
\newblock In \emph{Annual Conference on Advances in Neural Information
  Processing Systems (NeurIPS)}, 2020.

\bibitem[Heule et~al.(2013)Heule, Nunkesser, and
  Hall]{hyperloglogpractice:EDBT2013}
S.~Heule, M.~Nunkesser, and A.~Hall.
\newblock {HyperLogLog} in practice: Algorithmic engineering of a state of the
  art cardinality estimation algorithm.
\newblock In \emph{EDBT}, 2013.

\bibitem[Indyk and Motwani(1998)]{IndykMotwani:stoc98}
P.~Indyk and R.~Motwani.
\newblock Approximate nearest neighbors: Towards removing the curse of
  dimensionality.
\newblock In \emph{Proc. 30th Annual ACM Symposium on Theory of Computing},
  pages 604--613. ACM, 1998.

\bibitem[Ioannidis(2005)]{Ioannidis:2005}
John P.~A. Ioannidis.
\newblock Why most published research findings are false.
\newblock \emph{PLoS Med}, \penalty0 (2):\penalty0 8, 2005.

\bibitem[Janson(2017{\natexlab{a}})]{JansonTailbounds:2017}
Svante Janson.
\newblock Tail bounds for sums of geometric and exponential variables,
  2017{\natexlab{a}}.
\newblock URL \url{https://arxiv.org/abs/1709.08157}.

\bibitem[Janson(2017{\natexlab{b}})]{janson2017tail}
Svante Janson.
\newblock Tail bounds for sums of geometric and exponential variables,
  2017{\natexlab{b}}.
\newblock URL \url{https://arxiv.org/abs/1709.08157}.

\bibitem[Jayaram and Woodruff(2023)]{JayaramWoodruff:TALG2023}
Rajesh Jayaram and David~P. Woodruff.
\newblock Towards optimal moment estimation in streaming and distributed
  models.
\newblock \emph{ACM Trans. Algorithms}, 19\penalty0 (3), jun 2023.
\newblock ISSN 1549-6325.
\newblock \doi{10.1145/3596494}.
\newblock URL \url{https://doi.org/10.1145/3596494}.

\bibitem[Kane et~al.(2010)Kane, Nelson, and Woodruff]{KNW:PODS2010}
D.~M. Kane, J.~Nelson, and D.~P. Woodruff.
\newblock An optimal algorithm for the distinct elements problem.
\newblock In \emph{PODS}, 2010.

\bibitem[Knop and Steinke(2023)]{SteinkKnop:2023}
Alexander Knop and Thomas Steinke.
\newblock Counting distinct elements under person-level differential privacy.
\newblock \emph{CoRR}, abs/2308.12947, 2023.
\newblock \doi{10.48550/ARXIV.2308.12947}.
\newblock URL \url{https://doi.org/10.48550/arXiv.2308.12947}.

\bibitem[Knuth(1998)]{Knuth2}
D.~E. Knuth.
\newblock \emph{The Art of Computer Programming, Vol 2, Seminumerical
  Algorithms}.
\newblock Addison-Wesley, 2nd edition, 1998.

\bibitem[Kreuter et~al.(2020)Kreuter, Wright, Skvortsov, Mirisola, and
  Wang]{49177}
Benjamin Kreuter, Craig~William Wright, Evgeny~Sergeevich Skvortsov, Raimundo
  Mirisola, and Yao Wang.
\newblock Privacy-preserving secure cardinality and frequency estimation.
\newblock Technical report, Google, LLC, 2020.

\bibitem[Lukacs et~al.(2009)Lukacs, Burnham, and
  Anderson]{FreedmanParadox:2009}
Paul~M. Lukacs, Kenneth~P. Burnham, and David~R. Anderson.
\newblock Model selection bias and {F}reedman's paradox.
\newblock \emph{Annals of the Institute of Statistical Mathematics},
  62\penalty0 (1):\penalty0 117, 2009.

\bibitem[Mironov et~al.(2008)Mironov, Naor, and Segev]{MironovNS:STOC2008}
Ilya Mironov, Moni Naor, and Gil Segev.
\newblock Sketching in adversarial environments.
\newblock In \emph{Proceedings of the Fortieth Annual ACM Symposium on Theory
  of Computing}, STOC '08, page 651–660, New York, NY, USA, 2008. Association
  for Computing Machinery.
\newblock ISBN 9781605580470.
\newblock \doi{10.1145/1374376.1374471}.
\newblock URL \url{https://doi.org/10.1145/1374376.1374471}.

\bibitem[Nelson et~al.(2014)Nelson, Nguy$\tilde{\hat{\mbox{e}}}$n, and
  Woodruff]{NelsonNW13}
Jelani Nelson, Huy~L. Nguy$\tilde{\hat{\mbox{e}}}$n, and David~P. Woodruff.
\newblock On deterministic sketching and streaming for sparse recovery and norm
  estimation.
\newblock \emph{Lin. Alg. Appl.}, 441:\penalty0 152--167, January 2014.
\newblock Preliminary version in RANDOM 2012.

\bibitem[Pagh and Stausholm(2021)]{paghS:ICDT2021}
Rasmus Pagh and Nina~Mesing Stausholm.
\newblock Efficient differeffentially private {F}$_0$ linear sketching.
\newblock In Ke~Yi and Zhewei Wei, editors, \emph{24th International Conference
  on Database Theory (ICDT 2021)}, volume 186 of \emph{Leibniz International
  Proceedings in Informatics (LIPIcs)}, pages 18:1--18:19, Dagstuhl, Germany,
  2021. Schloss Dagstuhl -- Leibniz-Zentrum f{\"u}r Informatik.
\newblock ISBN 978-3-95977-179-5.
\newblock \doi{10.4230/LIPIcs.ICDT.2021.18}.
\newblock URL
  \url{https://drops.dagstuhl.de/entities/document/10.4230/LIPIcs.ICDT.2021.18}.

\bibitem[Papernot et~al.(2017)Papernot, McDaniel, Goodfellow, Jha, Celik, and
  Swami]{papernot2017practical}
Nicolas Papernot, Patrick McDaniel, Ian Goodfellow, Somesh Jha, Z~Berkay Celik,
  and Ananthram Swami.
\newblock Practical black-box attacks against machine learning.
\newblock In \emph{Proceedings of the 2017 ACM on Asia conference on computer
  and communications security}, pages 506--519, 2017.

\bibitem[Paterson and Raynal(2021)]{cryptoeprint:2021/1139}
Kenneth~G. Paterson and Mathilde Raynal.
\newblock Hyperloglog: Exponentially bad in adversarial settings.
\newblock Cryptology ePrint Archive, Paper 2021/1139, 2021.
\newblock URL \url{https://eprint.iacr.org/2021/1139}.
\newblock \url{https://eprint.iacr.org/2021/1139}.

\bibitem[Pettie and Wang(2021)]{PettieW:STOC2021}
Seth Pettie and Dingyu Wang.
\newblock Information theoretic limits of cardinality estimation: Fisher meets
  shannon.
\newblock In Samir Khuller and Virginia~Vassilevska Williams, editors,
  \emph{{STOC} '21: 53rd Annual {ACM} {SIGACT} Symposium on Theory of
  Computing, Virtual Event, Italy, June 21-25, 2021}, pages 556--569. {ACM},
  2021.
\newblock \doi{10.1145/3406325.3451032}.
\newblock URL \url{https://doi.org/10.1145/3406325.3451032}.

\bibitem[Reviriego and Ting(2020)]{DBLP:journals/icl/ReviriegoT20}
Pedro Reviriego and Daniel Ting.
\newblock Security of hyperloglog {(HLL)} cardinality estimation:
  Vulnerabilities and protection.
\newblock \emph{{IEEE} Commun. Lett.}, 24\penalty0 (5):\penalty0 976--980,
  2020.
\newblock \doi{10.1109/LCOMM.2020.2972895}.
\newblock URL \url{https://doi.org/10.1109/LCOMM.2020.2972895}.

\bibitem[Ros{\'e}n(1997)]{Rosen1997a}
B.~Ros{\'e}n.
\newblock Asymptotic theory for order sampling.
\newblock \emph{J. Statistical Planning and Inference}, 62\penalty0
  (2):\penalty0 135--158, 1997.

\bibitem[Shiloach and Even(1981)]{ShiloachEven:JACM1981}
Yossi Shiloach and Shimon Even.
\newblock An on-line edge-deletion problem.
\newblock \emph{J. ACM}, 28\penalty0 (1):\penalty0 1–4, jan 1981.
\newblock ISSN 0004-5411.
\newblock \doi{10.1145/322234.322235}.
\newblock URL \url{https://doi.org/10.1145/322234.322235}.

\bibitem[Smith et~al.(2020)Smith, Song, and
  Guha~Thakurta]{NEURIPS2020_e3019767}
Adam Smith, Shuang Song, and Abhradeep Guha~Thakurta.
\newblock The flajolet-martin sketch itself preserves differential privacy:
  Private counting with minimal space.
\newblock In H.~Larochelle, M.~Ranzato, R.~Hadsell, M.F. Balcan, and H.~Lin,
  editors, \emph{Advances in Neural Information Processing Systems}, volume~33,
  pages 19561--19572. Curran Associates, Inc., 2020.
\newblock URL
  \url{https://proceedings.neurips.cc/paper_files/paper/2020/file/e3019767b1b23f82883c9850356b71d6-Paper.pdf}.

\bibitem[Steinke and Ullman(2015)]{SteinkeUllman:COLT2015}
Thomas Steinke and Jonathan Ullman.
\newblock Interactive fingerprinting codes and the hardness of preventing false
  discovery.
\newblock In Peter Grünwald, Elad Hazan, and Satyen Kale, editors,
  \emph{Proceedings of The 28th Conference on Learning Theory}, volume~40 of
  \emph{Proceedings of Machine Learning Research}, pages 1588--1628, Paris,
  France, 03--06 Jul 2015. PMLR.
\newblock URL \url{https://proceedings.mlr.press/v40/Steinke15.html}.

\bibitem[Szegedy et~al.(2013)Szegedy, Zaremba, Sutskever, Bruna, Erhan,
  Goodfellow, and Fergus]{szegedy2013intriguing}
Christian Szegedy, Wojciech Zaremba, Ilya Sutskever, Joan Bruna, Dumitru Erhan,
  Ian Goodfellow, and Rob Fergus.
\newblock Intriguing properties of neural networks.
\newblock \emph{arXiv preprint arXiv:1312.6199}, 2013.

\bibitem[Wajc(2020)]{Wajc:STOC2020}
David Wajc.
\newblock \emph{Rounding Dynamic Matchings against an Adaptive Adversary}.
\newblock Association for Computing Machinery, New York, NY, USA, 2020.
\newblock URL \url{https://doi.org/10.1145/3357713.3384258}.

\bibitem[Woodruff and Zhou(2021)]{WoodruffZ21}
David~P. Woodruff and Samson Zhou.
\newblock Tight bounds for adversarially robust streams and sliding windows via
  difference estimators.
\newblock In \emph{Proceedings of the 62nd {IEEE} Annual Symposium on
  Foundations of Computer Science (FOCS)}, 2021.

\end{thebibliography}

\onecolumn
\newpage
\appendix

\section{Analysis of the Attack on the Standard Estimators} \label{kminsattack:sec}
This section includes the proof of Theorem~\ref{stattackutility:thm}.  
We first consider
$k$-mins sketches and $T(S) = \|S\|_1$. The modification needed for bottom-$k$ sketches are in Section~\ref{standardbottomk:sec}.

\subsection{Preliminaries}

The following are order statistics properties useful for analysing MinHash sketches.
Let $X_i \sim \Exp[1]$ for $i\in[n]$ be i.i.d.\ random variables.
Then the distribution of the minimum value and of the differences between the $i+1$ and the $i$th order statistics (smallest values) are independent random variables with distributions
\begin{align*}
    \Delta_1 &:= \min_{i\in [n]} X_i \sim \Exp[n] \\
    \Delta_i &:= \{ X_i \}_{(i+1)} - \{ X_i \}_{(i)} \sim \Exp[n-i] & i>1
\end{align*}

\begin{lemma}[Chebyshev's Inequality] \label{chebyshev:lemma}
\[
\Pr\left[ |Z-\E[Z]| \leq c\sigma^2 \right] \leq 1/c^2\ .
\]
\end{lemma}

We set some notation:
For a fixed ground set $N$ and randomness $\rho$, 
for each hash function $i\in [k]$, let $m^i_j\in N$ be the key with the $j$th rank in the $i$th hashmap, that is, $h_i(m^i_j)$ is the $j$th smallest in 
$\{h_i(u)\}_{u\in N}$.
Let
\begin{align*}
L &:=\log_2(rk) + 10 \\
N^i_0 &:= \{m^i_j\}_{j\leq L}\\
N_0 &:= \bigcup_{i\in [k]} N^i_0\\
N' &:= N \setminus N_0  
\end{align*}
be a rank threshold $L$, for $i\in [k]$ the set $N^i_0$ of keys with rank up to $L$ in the $i$th hashmap, the set $N_0$ that is the union of these keys across hashmaps, and the set $N'$ of the remaining keys in $N$.

We show that a choice of $n=O(k L/\alpha)$ ensures that 
certain properties that simplify our analysis hold.  
Our analysis applies to the event that these properties are satisfied:
 \begin{lemma} [Good draws] \label{goodNrho:lemma}
For $n= \Omega(\frac{1}{\alpha} k \log(r k))$, the following hold with probability at least $0.99$:
\begin{itemize}
    \item [p1] (property of $\rho$ and $N$) The keys $m^i_j$ for $i\in [k]$ and $j\leq L$ are distinct.
    \item [p2] In a run of Algorithm~\ref{standardattack:algo}, all $r$ steps, for all $i\in [k]$, $U$ includes a key from $N^i_0$.
    \item [p3] $n \geq 3kL/\alpha$
\end{itemize}
\end{lemma}
\begin{proof}
p3 is immediate. For p1, note that if we set $n\geq \frac{2}{p} k L$ then the claim follows with probability $1-p$ using the birthday paradox.
For p2, the probability that a random $U$ does not include one of the $L$ smallest values of a particular hash function is $2^{-L}$. 
    The probability that this happens for any of the $k$ hash functions in any of the $r$ rounds is at most $r k 2^{-L}$. Substituting $L= \log(rk) + 10$ we get the claim.
\end{proof}

 For fixed $N$ and $\rho$, 
consider the random variable $Z := T(S_\rho(U))$ over sampling of $U$ and the contributions $Z_i$ of hash function $i\in[k]$ to $Z$.
\begin{align*}
     Z_i &:= \min_{x\in U} h_i(x) \\
     Z &:= \sum_{i\in [k]} Z_i
\end{align*}

For a key $u\in N$, we consider the random variables $Z_i \mid u\in U$ and $Z \mid u\in U$ that are conditioned on $u\in U$.
From property p1 in Lemma~\ref{goodNrho:lemma}, the $Z_i$ are independent and are also independent when conditioned on $u\in U$.

\subsection{Proof outline}
We will need the following two Lemmas (the proofs are deferred to Section~\ref{gapvarproofs:sec}). 
Intuitively we expect $\E_U[Z]$ to be lower when conditioned on $m^i_1\in U$. We bound this gap from below.
For fixed $\rho$ and $N$, let
\[
G(u,v) := \E_{U \mid u\in U}[Z] -\E_{U \mid v\in U}[Z]\ .
\]
\begin{lemma} [Expectations gap bound] \label{gap:lemma}
For each $i\in [k]$ and $\delta>0$, 
\[
\Pr_{\rho,N}\left[\min_{u\in N'} G(u,m^i_1) \geq  \frac{\delta}{3n}\right] \geq 1-\delta\ .
\]
\end{lemma}

We bound from above the maximum over $u\in N$ of $\Var_{U \mid u\in U}[Z]$:
\begin{lemma}[Variance bound]\label{varbound:lemma}
For $\delta>0$ there is a constant $c$, 
\[
     \Pr_{\rho,N}\left[\max_{u\in N} \Var_{U \mid u\in U}[Z] \leq c\left(1+\frac{1}{\sqrt{k\delta}}\right) \frac{k}{n^2}\right]  \geq 1-\delta\ .
 \]
 \end{lemma}

We use the following to bound the number $r$ of attack queries needed so that the sorted order by average score  separates the minimum hash keys from the bulk of the keys in $N'$: 
\begin{lemma} [Separation] \label{shift:lemma}
    Let $\alpha>0$. Assume that 
    \begin{itemize}
        \item
        $\min_{u\in N'} G(u,m^i_1) \geq M > 0$
\item
$\max_{u\in N} \Var_{U \mid u\in U}[Z] \leq V^2$
\item
During Algorithm~\ref{standardattack:algo}, the keys $u\in N'$ and $m^i_1$ are selected in $U$ in at least
$r' \geq \frac{2 V^2}{M^2} \frac{1}{\alpha}$ rounds each.
    \end{itemize}
Then \[
\Pr[A[u] > A[m^i_1]] \leq \alpha
\]
\end{lemma}
\begin{proof}
Consider the random variable
\[
Y=A[u]-A[m^i_1]\ .
\]
From our assumptions:
\begin{align*}
    \E[Y] &\geq M  \\
    \Var[Y] &\leq \frac{2 V^2}{r'}
\end{align*}
We get 
\begin{align*}
    \Pr[Y<0] &\leq \Pr[ |Y-\E[Y]| \geq E[Y]] \\
&\leq \Pr[ |Y-\E[Y]| \geq M] \\
&= \Pr[ |Y-\E[Y]| \geq \frac{1}{\sqrt{\alpha}}\cdot \frac{V}{\sqrt{2r'}}] \leq \alpha
\end{align*}
Using Chebyshev's inequality.
\end{proof}

We are now ready to conclude the utility proof of Algorithm~\ref{standardattack:algo}:

 \begin{lemma} [Utility of Algorithm~\ref{standardattack:algo}] \label{stkminsutility:thm}
For $\alpha,\delta>0$. Consider Algorithm~\ref{standardattack:algo} with $r=O(\frac{k}{\delta \alpha})$. Then for each $i\in[r]$, with probability $1-\delta$, for any $u\in N'$ we have
$\Pr[A[u] < A[m^i_1]] \leq \alpha$.
\end{lemma}

\begin{proof}
Consider a key $m^i_1$ and a key 
$u\not\in N'$. We bound the probability that $A[u] < A[m^i_1]$. 
 
From Lemma~\ref{varbound:lemma}, with probability $1- 1/k$ we have  $V^2 = O\left(\frac{k}{n^2}\right)$.  From Lemma~\ref{gap:lemma}, for each $i$, with probability $1-\delta$, we have
$M\geq \delta/(3n)$. If we choose $r = 4r'$ in Algorithm~\ref{standardattack:algo}, then with high probability a key $x\in N$ is selected to $U$ at least $r'$ times. The claim follows from Lemma~\ref{shift:lemma} by setting
$r'= O(\frac{k}{\delta \alpha})$ .
\end{proof}

We are now ready to conclude the proof of Theorem~\ref{stattackutility:thm}.
Recall that a subset $U$ of keys of size $M=2\alpha n$ has $T(S(U))$ over random $\rho$ concentrated around 
$k/M$ with standard error $\sqrt{k}/M)$.
We now consider the prefix $U$ of $M=2\alpha n$ keys in the sorted order by average scores.
This selection 
has 
$\E[T(S(U))] \leq (3\alpha)  \cdot k/(2\alpha n)$. To see this, 
note that $(1-\delta)$ of the MinHash values are the same as in the sketch of $N$. Therefore they have expected value $1/n$. The remaining $\delta$ fraction have expected value $1/(2\alpha n)$. Therefore $\E[T_\rho(S(U))] = k((1-\delta)/n + \delta/(2\alpha n)) = (k/(2\alpha n))(\delta + (1-\delta)2\alpha)$. Therefore $T$ is a factor of $1/(\delta + 2\alpha)$ too small, which for small $\alpha$ and $\delta < \alpha$ is a large constant multiplicative error.

\subsection{Proofs of Lemma~\ref{gap:lemma} and Lemma~\ref{varbound:lemma}} \label{gapvarproofs:sec}

For fixed $\rho$ (and $N$), for $i\in [k]$ and $j\in [N-1]$, denote by
\begin{align*}
    \Delta^i_1 &:= \min_{x\in N} h_i(x) \equiv h_i(m^i_1) \\
    \Delta^{i}_j &:= \{h_i(x) \mid x\in N\}_{(j+1)} - \{h_i(x) \mid x\in N\}_{(j)} \equiv h_i(m^i_{j+1}) -h_i(m^i_{j}) & \text{ for $j>1$}
\end{align*}    
the gap between the $j$ and $j+1$ smallest values in $\{h_i(x)\}_{x\in N}$.

\begin{lemma} [Properties of $\Delta^i_j$] \label{Deltaprop:lemma}
  The random variables $\Delta^i_j$ $i\in [k]$, $j\in [L]$ over the sampling of $N$ are independent with distributions $\Delta^{i}_j \sim \Exp[N-j]$. This holds also when conditioning on properties p1 and p2 of Lemma~\ref{goodNrho:lemma}. 
\end{lemma}
\begin{proof}
It follows from properties of the exponential distribution that
over the sampling of $\rho$, $\Delta^{i}_j \sim \Exp[N-j]$ are independent random variables for $i,j$. Note that p1 and p2 are independent of the actual values of $h_i(m^i_j)$ and only depend on the rank in the order.
\end{proof}

We can now express the distribution of the random variable $Z_i$ in terms of $\Delta^i_j$:  We have 

For $j\geq 1$, the probability that $Z_i= \sum_{\ell=1}^j \Delta^i_\ell$ is $2^{-j}/(1-2^{-L})$.  This corresponds to the event that $U$ does not include the keys $m^i_\ell$ for $\ell<j$ and includes the key $m^i_j$.  The normalizing factor $(1-2^{-L})$ arises from property p2 in Lemma~\ref{goodNrho:lemma}. In the sequel we omit this normalizing factor for brevity.

\begin{align*}
Z_i = \begin{cases}
    \Delta^i_1 & \text{with probability } 2^{-1}/(1-2^{-L}) \\
    \Delta^i_1 + \Delta^i_2 & \text{with probability } 2^{-2}/(1-2^{-L}) \\
    \Delta^i_1 + \Delta^i_2 + \Delta^i_3 & \text{with probability } 2^{-3}/(1-2^{-L}) \\
    & \vdots \\
    \Delta^i_1 + \cdots + \Delta^i_j & \text{with probability } 2^{-j}/(1-2^{-L}) \\
    & \vdots
\end{cases}
\end{align*}

We now consider $Z_i$ conditioned on the event $m^i_j\in U$. Clearly for $j=1$ (conditioning on the event $m^i_1 \in U$) we have $Z_i \equiv \Delta^i_1$.
For $j\geq 1$, we have $Z_i=\sum_{\ell=1}^h \Delta^i_\ell$ with probability $2^{-h}$ for $h< j$ and 
$Z_i=\sum_{\ell=1}^j \Delta^i_\ell$ with probability $2^{-j+1}$.

We bound the expected value of $Z_i$ conditioned on presence of a key $u\in U$.
 \begin{lemma} \label{gapexpectationhelp:lemma}
 \begin{itemize}
     \item[(i)] (Anti-concentration)
 For $u\not= m^i_1$, the random variable over sampling of $\rho,N$
\[
G = \max_{u\in N\setminus\{m^i_1\}}\E_{U\mid u\in U}[Z_i] - \E_{U \mid m^i_1\in U}[Z_i]
\]
is such that for all $c >0$,
$\Pr_{\rho,N}\left[G \geq \frac{c}{n-1}\right
] \geq e^{-2c}$.
\item[(ii)] (Concentration)
% Bound the potential gain
For $u\in N$, the random variable
\[
G = \E_{U \mid u\in U}[Z_i] - \E_{U \mid m^i_j\in U}[Z_i]
\]
is such that for $c \geq 1$
\[
\Pr_{\rho,N}\left[G \geq c\cdot \frac{3}{n 2^{j}}\right] \leq c e^{-c}\ .
\]
\end{itemize}
\end{lemma}
\begin{proof}
Per Lemma~\ref{goodNrho:lemma} we are assuming the event (that happens with probability $1-\delta$ that $U$ includes a key $u \in \{m^i_j\}_{j\in [L]}$ for all $i\in [k]$. Therefore, for a key $u\in N'$ it holds that
 \[
 \E_{U \mid u \in U}[Z_i] =  \E_{U}[Z_i]\ .
 \]

Recall that $Z_i \mid m^i_1 \in U = \Delta^i_1$. Otherwise (when not conditioned on $m^i_1\in U$ or conditioned on presence of $u\not= m^i_1$) $Z_i=\Delta^i_1$ with probability $1/2$ (when the random $U$ includes $m^i_1$) and $Z_i \geq \Delta^i_1 + \Delta^i_2$ otherwise (when $U$ does not include $m^i_1$). Thus,
\begin{align*}
    \E_{U}[Z_i] - \E_{U \mid m^i_1\in U}[Z_i] \geq \Delta^i_2/2  \\
    \E_{U \mid m^i_j\in U}[Z_i]-\E_{U \mid m^i_1\in U}[Z_i] \geq 
    \Delta^i_2/2 & \text{ for $1<j\leq L$}
\end{align*}

Therefore for $u\not= m^i_1$, 
\[
G := \E_{U\mid u\in U}[Z_i] - \E_{U \mid m^i_1\in U}[Z_i] \geq \Delta^i_2/2\ .
\]
Since $\Delta^i_2 \sim \Exp[N-1]$,
we have for all $t>0$,
$\Pr_\rho[G \geq t] \geq e^{-2t/(n-1)}$. This establishes claim (i).  

For claim (ii), note that
\begin{align*}
\E_{U \mid u\in U}[Z_i] &\leq \E[Z_i] \leq \sum_{\ell=1}^L \frac{1}{2^{\ell-1}} \Delta^i_\ell \\
\E_{U \mid m^i_j\in U}[Z_i] &\geq \sum_{\ell=1}^{j-1} \frac{1}{2^{j-1}} \Delta^i_j \\
\end{align*}
% \edith{For clarity, ignore normalization factor by $1-2^{-L}$ from the conditioning. }
Therefore 
\[
G := \E_{U \mid u\in U}[Z_i] - \E_{U \mid m^i_j\in U}[Z_i] 
 \leq \sum_{\ell=j}^L \frac{1}{2^{\ell-1}} \Delta^i_\ell\ .
\]
This is a sum of independent exponential random variables and recall that we assumed $L\leq n/3$. Therefore,
this is stochastically smaller that the respective geometrically decreasing weighted sum of independent $\Exp[3/n]$ random variables. It follows that 
\[
\E\left[ G \right]\leq \frac{3}{n} \sum_{\ell=j}^L \frac{1}{2^{\ell-1}} = \frac{3}{n 2^{j}}
\]

We apply an upper bound on the tail~\citep{JansonTailbounds:2017} that shows that this concentrates almost as well as a single exponential random variable:
$\Pr[G \geq t \mu] \leq t e^{-t}$ and obtain claim (ii)
\[
\Pr\left[G \geq t \frac{3}{n 2^{j}}\right] \leq t e^{-t}\ .
\]
\end{proof}

Lemma~\ref{gap:lemma} is a corollary of the first claim of Lemma~\ref{gapexpectationhelp:lemma}.

We now express a bound on the variance of $Z_i$, also when conditioned on the presence of any key $u\in U$, for fixed $\rho$, $N$.
\begin{lemma} \label{expressvar:lemma}
    For fixed $N$, $\rho$, and any $u\in N$
    \[
    \Var_{U | u\in U}[Z_i], \Var[Z_i]=\Theta(\sum_{j} (3/2)^{-j} (\Delta_j^i)^2)\ .
    \]
\end{lemma}
\begin{proof}
\begin{align*}
    \Var_{U | u\in U}[Z_i], \Var[Z_i] &\leq \E[Z_i^2] \leq \sum_{j\geq 1} 2^{-j} \left(\sum_{\ell=1}^j \Delta^i_j\right)^2 \leq  \sum_{j\geq 1} 2^{-j} j \sum_{\ell=1}^j (\Delta^i_j)^2  \\
    &= \sum_{j\geq 1}  \left(\sum_{\ell\geq j} \frac{\ell}{2^{\ell}}\right) (\Delta^i_j)^2 =\Theta(\sum_{j} (3/2)^{-j} (\Delta_j^i)^2)
\end{align*}
\end{proof}

\begin{proofof}{Lemma~\ref{varbound:lemma}}
Since the $Z_i$, also when conditioned on $u\in U$, are independent (modulu our simplifying assumption in Lemma~\ref{goodNrho:lemma}), it follows from 
Lemma~\ref{expressvar:lemma} that
\[
    \Var_{U \mid u\in U}[Z]=\Theta\left(\sum_{i\in[k]} \sum_{j\geq 1} (3/2)^{-j} (\Delta_j^i)^2 \right)\ .
    \]
    Therefore,
    \[
    \max_{u\in N} \Var_{U \mid u\in U}[Z]=\Theta\left(\sum_{i\in[k]} \sum_{j\geq 1} (3/2)^{-j} (\Delta_j^i)^2 \right)\ .
    \]
The right hand side 
\[
Y := \max_{u\in N} \Var_{U \mid u\in U}[Z]
\]
is a random variable over $\rho,N$ that is a 
a weighted sum of the squares of independent exponential random variables $\Delta^i_j$.
The PDF of a squared exponential random variable $\Exp[w]^2$ is $\frac{w}{2\sqrt{t}}e^{-w \sqrt{t}}$. The 
mean is $\frac{2}{w^2}$ and the  variance is at most $\E[t^2] = 24/w^4$.  Applying this, we obtain 
that $\E[Y] = \Theta(k/n^2)$ and $\Var[Y] =  O(k/n^4)$.

From Chebyshev's inequality,
$\Pr[Y-\E[Y] \geq c \cdot \sqrt{k}/n^2] = O(1/c^2)$ and we obtain for any $\delta>0$ and a fixed constant $c$
 \[
     \Pr_{\rho,N}\left[\max_{u\in N} \Var_{U \mid u\in U}[Z] \geq c\left(1+\frac{1}{\sqrt{k\delta}}\right) \frac{k}{n^2}]\right]  \leq \delta\ .
 \]

\end{proofof}

\subsection{Attack on the Bottom-$k$ standard estimator}\label{standardbottomk:sec}

The argument is similar to that of $k$-mins sketches. We highlight the differences.
Recall that a bottom-$k$ sketch uses a single hash function $h$ with the sketch storing 
the $k$ smallest values $S(U) := \{h(x) \mid x\in U\}_{(1:k)}$. 
We use the $k$th order statistics ($k$th smallest value)
$T(S) := \{h(x) \mid x\in U\}_{(k)}$.

For fixed $\rho$ and $N$, let $m_j\in N$ ($j\in [n]$) be the key with the $j$th smallest hashmap  $h(m_j) = \{h(x) \mid x\in U\}_{(k)}$.
Define $\Delta_1 := h(m_1)$ and for $j>1$, $\Delta_j := h(m_{j}) - h(m_{j-1})$.

Let $R$ be the random variable that is the rank in $N$ of the key with the $k$th smallest hashmap in $U$. The distribution of $R$ is the sum of $k$ i.i.d.\ Geometric random variables $\Geom[q=1/2]$. 
We have $\E[R] = k/q$ and the concentration bound~\citep{JansonTailbounds:2017} that  for any $c\geq 1$, $\Pr[R > c \E[R]] \leq c e^{-c}$.  

Let $L = 10+ 2\log r$, $N_0 = \{m_j\}_{j\leq kL/q}$ be the keys with the $L (k/q)$ smallest hashmaps. Let $N' = N\setminus N_0$ be the remaining keys.  We show that the attack separates with probability $\alpha$ a key with one of the bottom-$k$ ranks and a key in $N'$.

Assume $n > 3 |N_0|$. Assume that we declare failure when a set $U$ selected by the algorithm does not contain $k$ keys from $N_0$.  The probability of such selection is
at most $Le^{-L} < 0.01/r$ and at most $0.01$ in all $r$ steps.

For fixed $\rho$ and $N$, consider the random variable $Z:= T(S(U))$. 
The following parallels Lemma~\ref{gap:lemma} and
Lemma~\ref{expressvar:lemma}:
\begin{lemma}
Fixing $\rho,N$,
\begin{itemize}
\item[(i)]
let
\[
G(u,v) := \E_{U \mid u\in U}[Z] -\E_{U \mid v\in U}[Z]
\]
For each $i\leq [k]$ and $\delta>0$, 
\[
\Pr\left[\min_{u\in N'} G(u,m_i) \geq  \frac{\delta}{3n}\right] \geq 1-\delta\ .
\]
\item [(ii)]
for $\delta>0$ and some constant $c$, 
\[
     \Pr_{\rho,N}\left[\max_{u\in N} \Var_{U \mid u\in U}[Z] \leq c\left(1+\frac{1}{\sqrt{k\delta}}\right) \frac{k}{n^2}\right]  \geq 1-\delta\ .
 \]
\end{itemize}
\end{lemma}
\begin{proof}
    (i) Note that $Z = m_R$. When $i\leq [k]$, $Z = m_{R+1}$. The gap is a weighted average of $\Delta_j$ for $j\in [R, |N_0|]$. These are independent $\Exp[n-i]$ random variables with $i\leq n/3$.  The expected value is $\Theta(1/n)$ and the tail bounds are at least as tight as for a single $\Exp[n/3]$ random variable.  

    (ii) We use the concentration bound on $R$ to express the variance for fixed $\rho,N$ as a weighted sum with total weight $\Theta(k)$ and each of weight $O(1)$ of independent squared exponential random variables.  The argument is as in the proof of Lemma~\ref{expressvar:lemma}.
\end{proof}

Using the same analysis, a subset $U\subset N$ of size $\alpha n$ has $T(S(U))$ that in expectation has the $k/\alpha$ smallest rank in $N$ with standard error $\sqrt{k}/\alpha$ and normalized standard error $1/\sqrt{k}$.  The subset $U$ selected as a prefix of the order generating in the attack includes the $(1-\delta)k$ of the bottom-$k$ in $N$ and 
$\delta k$ of the bottom in $U$.  This means that in expectation $T(S(U))$ has the $k \delta/\alpha$ rank in $N$.  That is, error that is $(1/\delta)$ factor off.

\section{Analysis of Attack on General Query Response Algorithms} \label{genattackproofs:sec}

We include details for Sections~\ref{singlebatch:sec} and~\ref{generalmain:sec}.

\subsection{Rank-domain representation of sketches}

We use the 
\emph{rank domain} representation $S_\rho^R(U)$ of the input sketch $S_\rho(U)$. This representation is defined for subsets of a fixed ground set $N$. Instead of hash values, it includes the ranks in $N$ of the keys that are represented in the sketch $S_\rho(U)$ with respect to the relevant hashmaps. 
\begin{definition} (Rank domain representation)
For a fixed ground set $N$, and a subset $U\subset N$, the 
\emph{rank domain} representation $S^R_\rho(U)$ of a respective MinHash sketch has   
the form
$(Y_1,\ldots,Y_k)$, where $Y_i \in \mathbbm{N}$.
\begin{itemize}
\item
    $k$-mins sketch: For $i\in[k]$ and $j\geq 1$, let $m^i_j$ be the key $x\in N$ with the $j$th smallest $h_i(x)$. 
For $i\in [k]$, let $Y_i := \arg\min_j m^i_j\in U$. That is, $Y_i$ is the smallest $j$ such that $m^i_j\in U$.
\item
$k$-partition sketch: For $i\in[k]$ and $j\geq 1$, let $m^i_j$ be the key $x\in N$ that is in part $i$ with the $j$th smallest $h(x)$. For $i\in [k]$, let $Y_i := \arg\min_j m^i_j\in U$ be the smallest rank in the $i$th part.
\item
Bottom-$k$ sketch: For $j\geq 1$, let $m_j$ be the key $x\in N$ with the $j$th smallest $h(x)$ value.  Let the bottom-$k$ keys in $U$ be $m_{i_1},m_{i_2},\ldots, m_{i_k}$ where $i_1<i_2<\cdots <i_k$. We then define
$Y_1 := i_i$ and $Y_j := i_j - i_{j-1}$ for $1<j\leq k$.
\end{itemize}
\end{definition}
Note that when the ground set $N$ is available to the query response algorithm the rank domain and MinHash representations are equivalent (we can compute one from the other).

The following properties of the rank domain facilitate a simpler analysis: (i)~It only depends on the order induced by the hashmaps and not on actual values and thus allows us to factor out dependence on $\rho$, (ii)~It  subsumes the information on $q$ (and $|U|$) available from $S_\rho(U)$ and (iii)~It has a unified form and facilitates a unified treatment across the MinHash sketch types. 

The subsets $U\sim \mathcal{D}_0$ generated by our attack algorithm selects a rate $q$ and then sample $U$ by including each $x\in N$ independently with probability $q$. 
We consider the distribution, which we denote by $S^R[q]$ of the rank domain sketch under this sampling of $U$ with rate $q$.
We show that for a sufficiently large $|N|=n$, the rank domain representation is as follows: 

\begin{lemma} [distribution of rank-domain sketches] \label{dist:lemma}
For $\delta>0$, $q\in (0,1)$,  and an upper bound $r$ on the attack size, let $L =\log_2(rk/\delta)/q + 10$, and assume $n > 3kL/\delta$.
Then for all the three MinHash sketch types, the distribution 
$S^R[q]$ is within total variation distance $\delta$ from 
$(Y_1,\ldots,Y_k)$ that are  
$k$ independent geometric random variables with parameter $q$: 
$Y_i \sim \Geom[q]$.
\end{lemma}
\begin{proof}
As in Lemma~\ref{goodNrho:lemma}. Applying the birthday paradox with $n>3kL/\delta$, with probability at least $1-\delta$:
For $k$-mins sketches, the keys $m^i_j$ for $i\in [k]$ and $j\in [L]$ are distinct. For $k$ partition sketches, there are at least $L$ keys assigned to each part so the keys $m^i_j$ for $i\in [k]$ and $j\in [L]$ are well specified.

A sketch from $S^R[q]$ can be equivalently sampled using the following process:
\begin{itemize}
\item
    $k$-mins and $k$-partition sketch: For each $i\in[k]$, process keys $m^i_j$ by increasing $j\geq 1$ until $\Bern[q]$ and then set $Y_i=j$.
\item Bottom-$k$ sketch:  Process keys $m_j$ in increasing $j\geq 1$ until we get $\Bern[q]$ $k$ times. 
\end{itemize}    

We next establish that with our choice of $n$, with probability at least $1-\delta$, in all of $r$ sampling of $U$, 
the sketch $S^R(U)$ is determined by the $Lk$ smallest rank keys. Therefore there are sufficiently many keys for the sketch to agree with the sampling $k$ i.i.d.\ $\Geom[k]$ random variables.

For $k$-mins and $k$-partition sketches, the probability that for a single hashmap $i\in [k]$ none of the $L$ smallest rank is included is at most $(1-q)^L$. Taking a union bounds over $k$ maps and $r$ sampling and using that $\log(1/(1-q) \approx q$ gives the claim. With bottom-$k$ sketches the requirement is that in all $r$ selections, the $k$th smallest rank is $O(kL)$.
\end{proof}

\begin{remark} \label{statisticprop:rem}
Estimating $q$ from a sketch from $S^R[q]$ is a standard parameter estimation problem. A sufficient statistic $T$ for estimating $q$  is $T(S^R) := \sum_{i=1}^k Y_i$. 
 Note the following properties:
\begin{itemize}
    \item 
    The distribution $S^R[q]$ does not provide  additional information on the cardinality $|U|$ beyond an estimate of $q$.
   \item
 The distribution of $S^R[q]$ conditioned on $T(S)=\tau$ is the same for all $q$ (this follows from the definition of sufficient statistic).
   \item
  The statistic $T$ has expected value $k/q$, variance   $k(1-q)/q^2$, and single-exponential concentration~\citep{janson2017tail}. 
 \end{itemize}
\end{remark}

\subsubsection{Continuous rank domain representation}
We now cast the distribution of $S^R[q]$ using a continuous representation $S^C$.  This is simply a tool we use in the analysis.

We can sample a sketch from $S^R[q]$ as follows
\begin{itemize}
    \item 
     Set the rate $q':=-\ln(1-q)$
    \item 
     Sample a sketch $S^C(U)=(Y'_1,\ldots,Y'_k)$ where $Y'_i \sim \Exp[q']$ are i.i.d
    \item 
     Compute $S^R(U)$ from $S^C(U)$ using $Y_i \gets  
1+\lfloor Y'_i \rfloor +1$ for $i\in [k]$.
\end{itemize}
The correctness of this transformation is from the relation between a geometric $\Geom[q]$ and exponential $\Exp[q']$ distributions:
\[
\Pr[Y_i=t] = \Pr[ t-1 \leq Y'_i < t] = e^{-q'(t-1)} - e^{-q't} = (1-e^{-q'})\cdot e^{- q't} = q\cdot(1-q)^t\ .
\]

Note that we can always recover $S^R$ from $S^C$ but we need to know $q$ in order to compute $S^C$ from $S^R$:
\[Y_i'\sim \Exp[-\ln(1-q)]\mid Y'_i\in [Y_i-1,Y_i)\ .\] 
Therefore being provided with the continuous representation only makes the query response algorithm more informed and potentially more powerful.  Also note that $|q-q'|< q^2/2$.

A sufficient statistic for estimating $q'$ from $S^C[q']$ is $T' := \sum_{i=1}^k Y'_i$.  
In the sequel we will work with $S^C$ and omit the prime from $q$ and $T$.

Now note that the distribution of $T := \| S^C[q]\|_1$ for a given $k$ and $q$ is
the sum of $k$ i.i.d.\ $\Exp[q]$ random variables. This is  the Erlang distribution that has density function for $x\in[0,\infty]$:
\begin{equation}\label{Erlangpdf:eq}
 f_T(k,q;x) = \frac{q^k}{(k-1)!} x^{k-1}e^{-qx}
\end{equation}
The distribution has mean $\E[T]=k/q$, variance 
$\Var[T] = k/q^2$ and exponential tail bounds~\citep{JansonTailbounds:2017}:
\begin{align}
    \text{For $c>1$:}\, & \Pr[T\geq c\cdot k/q] \leq \frac{1}{c} e^{-k(c-1-\ln c)} \label{tailup:eq}\\
    \text{For $c<1$:}\, & \Pr[T\leq c\cdot k/q] \leq  e^{-k(c-1-\ln c)} \label{taildown:eq}
\end{align}

Consider the random variable 
\begin{equation} \label{Z:eq}
    Z = (T-\E[T])/\sqrt{\Var[T]}
\end{equation}
 that is the number of standard deviations of $T$ from its mean. We have
$T= \frac{k}{q} + Z\cdot \frac{\sqrt{k}}{q}$ and $Z = \frac{qT}{\sqrt{k}}-\sqrt{k}$.

The domain of $Z$ is $[-\sqrt{k},\infty)$ and the density function of $Z$ is
\begin{align}
 f_Z(k;z) &= \frac{\sqrt{k}}{q}\frac{q^k}{(k-1)!}(\frac{k}{q}+z\frac{\sqrt{k}}{q})^{k-1}e^{-q(\frac{k}{q}+z\frac{\sqrt{k}}{q})}\nonumber\\
 &= \frac{\sqrt{k}}{(k-1)!}(k+z\sqrt{k})^{k-1}e^{-(k+z\sqrt{k})}\label{Deltapdf:eq}
\end{align}

This density satisfies 
\begin{align}
    &\int_{-\sqrt{k}}^\infty f_Z(k;x) x dx = 0 \label{EDelta:eq}\\
    &\int_{-\sqrt{k}}^\infty f_Z(k;x) x^2 dx = 1 \label{squared:eq} \\
    &\int_{-\sqrt{k}/4}^{\sqrt{k}/4} f_Z(k;x) x^2 dx = \Theta(1) \label{squaredin:eq} \\
    &\text{for $c\in (0,1]$, }\int_{-c}^0 f_Z(k;x) dx, \int_{0}^{c} f_Z(k;x) dx = \Omega(c) \label{boundedrange:eq}\\
    &\text{for $c \geq 0$, } \Pr[T \geq c\cdot \sqrt{k}] \leq \frac{1}{c+1} e^{-k\cdot (c-\ln(c+1))} \label{uptailZ:eq}\\
    &\text{for $c\in (0,1)$, } \Pr[T \leq -c\cdot \sqrt{k}] \leq e^{-k\cdot (1-c-\ln(1-c))} \label{lowtailZ:eq}
\end{align}
Note that $T$ is available to the query response algorithm but $q$, and thus the value of $Z$ are not available.

\subsection{Correct maps}

A map $S \mapsto \pi(S)\in [0,1]$ maps sketches to the probability of returning $1$. We require that the maps selected by QR are correct as in Definition~\ref{correctmap:def} with $\delta=O(1/\sqrt{k})$.

For a map $\pi$ and $\tau$ we denote by 
$\overline{\pi}(\tau)$ the mean value of $\pi(S)$ over sketches with statistic value $T(S)=\tau$. This is well defined since for the query distribution in our attacks $\mathcal{D}_0\mid q=q^*$, even when conditioned on a fixed rate $q^*$, 
the distribution of the sketch conditioned on $T(S)=\tau$ does not depend on $q^*$ (See Remark~\ref{statisticprop:rem}). 

We now specify conditions on the map
$\overline{\pi}(\tau)$ that must be satisfied by a correct $\pi$.
A correct map may return an incorrect output, when conditioned on cardinality, with probability $\delta$. This means that there are correct maps with large error on certain $\tau$ (since each cardinality has a distribution on $\tau$).  We therefore can not make a sharp claim on $\overline{\pi}(\tau)$ that must hold for any $\tau$ in an applicable range. Instead, we make an average claim: For any interval of $\tau$ values that is wide enough to include $\Omega(1)$ of the values for some cardinality value $c\not\in [A,2A]$, the average error of the mapping must be $O(\delta)$.

\begin{claim} \label{Tforcorrectmap:claim}
For any $\xi > 0$, there is $c_0>0$ such that for any correct map $\pi$ for $A$ and $\delta\leq c_0/\sqrt{k}$ and 
$\tau_b > (1+0.1/\sqrt{k}) \tau_a$ it holds that
\begin{equation}
\begin{cases}
  \frac{1}{\tau_b-\tau_a}  \int_{\tau_a}^{\tau_b} \overline{\pi}(x) dx < \xi & \text{if } \tau_a > \frac{kn}{A} (1-1/\sqrt{k}) \\
   \frac{1}{\tau_b-\tau_a} \int_{\tau(1-a/\sqrt{k})}^\tau \overline{\pi}(x) dx > 1- \xi & \text{if } \tau_b < \frac{kn}{2A}(1+1/\sqrt{k})
\end{cases}
\end{equation}
\end{claim}
\begin{proof}
For a cardinality value $c$, 
the distribution of the statistic $T$ conditioned on a cardinality value $c$ is $f_T(k,c/n;x)$ \eqref{Erlangpdf:eq}. With cardinality value $c$, it holds that
$\Pr[T< kn/c] \geq 1/e$ and $\Pr[T> kn/c] \geq 1/e$. Moreover, the density in the interval $\frac{kn}{c}[1-0.1/\sqrt{k},1+0.1/\sqrt{k}]$ is $\Theta(1)$.
It follows from the correctness requirement for cardinality value $c = k/\tau$ that there exists $c_1>0$ such that:
\begin{equation}
\begin{cases}
    \int_\tau^{\tau(1+0.1/\sqrt{k})} \overline{\pi}(x) dx <c_1 \delta & \text{if } \tau > \frac{kn}{A}(1-1/\sqrt{k}) \\
    \int_{\tau(1-0.1/\sqrt{k})}^\tau \overline{\pi}(x) dx > 1- c_1\delta & \text{if } \tau < \frac{kn}{2A}(1+1/\sqrt{k})
\end{cases}
\end{equation}
Therefore for $\tau_a > \frac{kn}{A} (1-1/\sqrt{k})$
\begin{align*}
    \int_{\tau_a}^{\tau_b} \overline{\pi}(x) dx \leq 10 \sqrt{k}(\tau_b-\tau_a) c_1 \delta
\end{align*}
and for $\tau_b < \frac{kn}{2A}(1+1/\sqrt{k})$
\begin{align*}
    \int_{\tau_a}^{\tau_b} \overline{\pi}(x) dx \geq (\tau_b-\tau_a)\cdot (1 - 10 \sqrt{k} c_1 \delta)\ .
\end{align*}
    Choosing $c_0 \leq 10 c_1$ establishes the claim.
\end{proof}

For fixed $q$, the cardinality $|U|$ of the selected $U$ has distribution $\Binom(q,n)$.
The $n$ chosen for the attack is large enough so that for all our $r$ queries $||U|-qn|/(qn) \ll 1/\sqrt{k}$. That is, the variation in $|U|$ for fixed $q$ is small compared with the error of the sketch and $|U|\approx qn$.

 \subsection{Relating $Z$ and sampling probability of low rank keys}

The sketch is determined by $k$ keys that are lowest rank in $U$. We can view the sampling of $U$ to the point that the sketch is determined in terms of a process, as in the proof of Lemma~\ref{dist:lemma}, that examines keys from the ground set $N$ in a certain order until the $k$ that determine the sketch are selected. The process selects each examined key with probability $q$. 
For a bottom-$k$ sketch, keys are examined in order of increasing rank until $k$ are selected. With $k$-mins and $k$-partition, keys in each part (or hash map) are examined sequentially by increasing rank until there is selection for the part. For all sketch types,
the statistic value $T$ corresponds to the number of keys from $N$ that are examined until $k$ are selected. This applies also with the continuous representation $S^C$.

We denote by $N_0$ the set of keys that are examined with probability at least $\delta_c \leq 1/(rk)$ when the rate is at least $q_a$. We refer to these keys as \emph{low rank} keys. It holds that
$|N_0| \leq k \ln(1/\delta_c)/q_a$.
For $\delta_c = 1/O(rk)$ we have $|N_0|=O(k \log(rk))$.
The remaining keys $N' := N\setminus N_0$ are unlikely to impact the sketch content and we refer to them as \emph{transparent}. 

With rate $q$, the probability that a certain key is included in $U$ is $q$. 
We now consider a rate $q$ and the inclusion probability \emph{conditioned on the normalized deviation from the mean $Z$}.
Transparent keys have inclusion probability $q$.
The low rank keys $N_0$ however have \emph{average} inclusion probability that depends on $Z$. Qualitatively, we expect that when $Z < 0$, the inclusion probability is larger than $q$ and this increases with magnitude $|Z|$. When $Z>0$, the inclusion is lower and decreases with the magnitude.  This is  quantified in the following claim:
\begin{claim} \label{AveM:claim}
    Fix a rate $q$ and $\Delta$. Consider the  distribution of $U$ conditioned on $Z=\Delta$. The average probability over $N_0$ keys to be in $U$ is $q - \Delta \frac{\sqrt{k}}{|N_0|}$.
\end{claim}
\begin{proof}
Equivalently, consider the distribution conditioned on $T= \frac{1}{q}(k+\Delta\sqrt{k})$. % assuming $T < |N_0|$. 
The sampling process selects $k$ keys after examining $T$ keys.  The average effective sampling rate for the examined keys is $q_{e} = k/T$. 
There are $T$ keys out of $N_0$ that are processed with effective rate $q_e$ and the remaining keys in $N_0$ have effective rate $q$. 

Averaging the effective rate over the $T= \frac{1}{q}(k+Z\sqrt{k})$ processed keys and the remaining $N_0$ keys we obtain
\begin{align*}
    \frac{T\cdot q_e + (|N_0|-T)\cdot q}{|N_0|} &=
 \frac{T\cdot \frac{k}{T} + (|N_0|-T)\cdot q}{|N_0|}   = \frac{k + (|N_0|-(\frac{k + \Delta\sqrt{k}}{q}))\cdot q}{|N_0|}\nonumber \\ &= q - \Delta \frac{\sqrt{k}}{|N_0|} 
\end{align*}

\end{proof}

\subsection{Scoring probability gap}

For a map $\pi$, let $p'(\pi)$ be the score probability, over the distribution of $q$ and $Z$, of a key in $N'$. Let $p_0(\pi)$ be the \emph{average} over $N_0$ of the score probability of keys in $N_0$.

Let $f_\lambda(x)$ be the density function of the selected rate, described by Algorithm~\ref{samplerate:algo}. 
\begin{algorithm2e}[h]\caption{\small{\texttt{Sample rate}}}\label{samplerate:algo}
{\small
\DontPrintSemicolon
\KwIn{$A$,  $n$}
$\omega \gets \frac{n}{2A}$\;
$\omega_a\gets \frac{1}{2}\omega$; $\omega_b \gets \frac{5}{2} \omega$\tcp*{range of inverse rates}
$D\sim U[0,\omega/4]$\;
$\omega^*_a  \gets \omega_a+D$; $\omega^*_b \gets \omega^*_a + \frac{7}{4}\omega$\tcp*{range of sampled inverse rate}
\Return{$q\sim \frac{1}{U[\omega^*_a,\omega^*_b]}$}
}
\end{algorithm2e}
Note that the selected rate is in the interval
$q\in [\frac{1}{\omega_b},\frac{1}{\omega_a}]= \frac{1}{\omega}\cdot [\frac{2}{5},2] =\frac{A}{n}\cdot[\frac{4}{5},4]$.

 For each transparent key, the score probability is:
 \begin{align}
     p'(\pi) = \int_{q_a}^{q_b} \int_{-\sqrt{k}}^\infty \overline{\pi}(\frac{k}{q}(1+z/\sqrt{k})) f_Z(k; z) \cdot d z\cdot q \cdot f_\lambda(q) dq \label{scoreprobGen:eq}
 \end{align}
 
 On average over the low-rank keys $N_0$ using Claim~\ref{AveM:claim}  it is 
  \begin{align}
     p_0(\pi) = \int_{q_a}^{q_b} \int_{-\sqrt{k}}^\infty  \overline{\pi}(\frac{k}{q}(1+z/\sqrt{k})) f_Z(k; z) \cdot \left(q - z \frac{\sqrt{k}}{|N_0|}\right)\cdot dz \cdot f_\lambda(q) dq \label{scoreprobM:eq}
 \end{align}

For a correct map $\pi$ (as in Definition~\ref{correctmap:def}), we express the gap between $p'(\pi)$ and $p_0(\pi)$. Note that we bound the gap without assuming much on the actual values, as they can highly vary for  different correct $\pi$.
\begin{lemma} [Score probability gap] \label{scoregap:lemma}
Consider a step of the algorithm and a correct map $\pi$ (see Remark~\ref{manycorrect:rem}).  
Then
\[
p_0(\pi)-p'(\pi) = \Omega\left(\frac{1}{|N_0|}\right) = \Omega\left(\frac{1}{k\log(kr)}\right)
\]
\end{lemma}

In the remaining part of this section we present the proof of Lemma~\ref{scoregap:lemma}. We will need the following claim, that relates $\Delta$ and scoring probability.
\begin{claim} \label{deltascore:claim}
For $|\Delta| < \sqrt{k}/4$
\[
\int_{q_a}^{q_b} \overline{\pi}(\frac{k}{q}(1+\frac{\Delta}{\sqrt{k}}))\cdot f_\lambda(q) dq = \int_{q_a}^{q_b} \overline{\pi}(\frac{k}{q})\cdot f_\lambda(q) dq + \Theta(\frac{\Delta}{\sqrt{k}})
\]
\end{claim}
\begin{proof}
Using the distribution specified in Algorithm~\ref{samplerate:algo}, for any $g()$:
\begin{align}
 \int_{q_a}^{q_b} g(q) f_\lambda(q) dq   = -\frac{4}{\omega}\int_0^{\omega/4} dD \frac{4}{7\omega} \int_{\omega_a+D}^{\omega_a+D+\frac{7}{4}\omega} g(1/x) dx \label{q2omega:eq}
\end{align}

We use $w^*_a = \omega_a+D$ and $w^*_b = \omega_a+D+\frac{7}{4}\omega$ and get
\begin{align*}
    \int_{\omega^*_a}^{\omega^*_b} \overline{\pi}(k x (1+\frac{\Delta}{\sqrt{k}}))\cdot d x &= \frac{1}{1+\frac{\Delta}{\sqrt{k}}}\int_{\omega^*_a \cdot(1+\Delta/\sqrt{k})}^{\omega^*_b\cdot (1+\Delta/\sqrt{k})} \overline{\pi}(k y) dy \text{$\;\;$ (change variable $x$ to $y=x (1+\Delta/\sqrt{k})$)} \nonumber\\
    &= \frac{1}{1+\frac{\Delta}{\sqrt{k}}}\left(\int_{\omega^*_a}^{\omega^*_b} \overline{\pi}(k x)dx -  \int_{\omega^*_a }^{\omega^*_a (1+\Delta/\sqrt{k})}\overline{\pi}(k x)dx + \int_{\omega^*_b }^{\omega^*_b (1+\Delta/\sqrt{k})}\overline{\pi}(k x)dx \right)
\end{align*}

Therefore,\footnote{Note that $\Delta$ can be negative. In which case in order to streamline expressions we interpret the asymptotic notation $O(c\Delta)$ as $-O(c|\Delta|)$.}
\begin{align}
    \lefteqn{\int_{\omega^*_a}^{\omega^*_b} \left( \overline{\pi}(k x (1+\frac{\Delta}{\sqrt{k}})) - \overline{\pi}(kx)\right) \cdot d x =}\label{diffomega:eq}\\
&=    -\Theta(\frac{\Delta}{\sqrt{k}})\cdot \int_{\omega^*_a}^{\omega^*_b} \overline{\pi}(kx) dx - \Theta(1)\cdot \int_{\omega^*_a }^{\omega^*_a (1+\Delta/\sqrt{k})}\overline{\pi}(k x)dx + \Theta(1)\cdot \int_{\omega^*_b }^{\omega^*_b (1+\Delta/\sqrt{k})}\overline{\pi}(k x)dx \nonumber\\
    &= -O(\frac{\Delta}{\sqrt{k}}\omega) - \Theta(1)\cdot \int_{\omega^*_a }^{\omega^*_a (1+\Delta/\sqrt{k})}\overline{\pi}(k x)dx + \Theta(1)\cdot \int_{\omega^*_b }^{\omega^*_b (1+\Delta/\sqrt{k})}\overline{\pi}(k x)dx\ .\nonumber
\end{align}  

The last equality follows using
\[\int_{\omega^*_a}^{\omega^*_b} \overline{\pi}(kx) dx \in [0, \omega^*_b-\omega^*_a] = [0, \frac{7}{4}\omega]\]

\begin{align}
\lefteqn{  \int_{q_a}^{q_b} \overline{\pi}(\frac{k}{q}(1+\frac{\Delta}{\sqrt{k}}))\cdot f_\lambda(q) dq - \int_{q_a}^{q_b} \overline{\pi}(\frac{k}{q})\cdot f_\lambda(q) dq =}\label{claimdiff:eq}\\
\lefteqn{=\int_{q_a}^{q_b} \left(\overline{\pi}(\frac{k}{q}(1+\frac{\Delta}{\sqrt{k}}))-\overline{\pi}(\frac{k}{q})\right)\cdot f_\lambda(q) dq =}\nonumber\\
&= \frac{4}{\omega}\int_0^{\omega/4} dD \frac{4}{7\omega} \int_{\omega_a+D}^{\omega_a+D+\frac{7}{4}\omega} \left(\overline{\pi}(k x (1+\frac{\Delta}{\sqrt{k}}))-\overline{\pi}(k x)\right) dx \;\;\text{(Using \eqref{q2omega:eq})}\nonumber\\
&= -O(\frac{\Delta}{\sqrt{k}})-\Theta(\frac{1}{\omega^2})\cdot \int_0^{\omega/4} dD \int_{\omega^*_a }^{\omega^*_a (1+\Delta/\sqrt{k})}\overline{\pi}(k x)dx + \Theta(\frac{1}{\omega^2})\cdot \int_0^{\omega/4} dD\cdot \int_{\omega^*_b }^{\omega^*_b (1+\Delta/\sqrt{k})}\overline{\pi}(k x)dx\ \;\;\text{(Apply \eqref{diffomega:eq})}\nonumber
\end{align}

We now separately bound terms\footnote{Argument for negative $\Delta$ is similar}:
\begin{align}
   \int_0^{\omega/4} dD \int_{\omega^*_a }^{\omega^*_a (1+\Delta/\sqrt{k})}\overline{\pi}(k x)dx &\geq
   \int_0^{\omega/4} dD \int_{\omega_a+D }^{\omega_a +D+ \frac{\Delta}{\sqrt{k}}\omega_a}\overline{\pi}(k x)dx \nonumber\\
   &= \int_{0}^{\frac{\Delta}{\sqrt{k}}(\omega/4)} dW \int_{\omega_a+W}^{\omega_a+\omega/4+W} \overline{\pi}(k x)dx\nonumber\\
   &\geq  \frac{\Delta}{\sqrt{k}}\frac{\omega}{4}\cdot \frac{\omega}{4}(1-\xi) = \Theta(\omega^2 \frac{\Delta}{\sqrt{k}}) \;\; \text{(Using Claim~\ref{Tforcorrectmap:claim})}\label{leftendlowbound:eq}\\
   \int_0^{\omega/4} dD \int_{\omega^*_a }^{\omega^*_a (1+\Delta/\sqrt{k})}\overline{\pi}(k x)dx &\leq
   \int_0^{\omega/4} dD \int_{\omega_a+D }^{\omega_a +D+ \frac{\Delta}{\sqrt{k}}(\omega_a+\omega/4)}\overline{\pi}(k x)dx \nonumber\\
   &= \int_{0}^{\frac{\Delta}{\sqrt{k}}(\omega_a+\omega/4)} dW \int_{\omega_a+W}^{\omega_a+\omega/4+W} \overline{\pi}(k x)dx
   &= O(\omega^2 \frac{\Delta}{\sqrt{k}})\label{leftendupbound:eq}
\end{align}
Combining \eqref{leftendlowbound:eq} and \eqref{leftendupbound:eq} we obtain that
\begin{equation} \label{leftendbound:eq}
  \int_0^{\omega/4} dD \int_{\omega^*_a }^{\omega^*_a (1+\Delta/\sqrt{k})}\overline{\pi}(k x)dx =  \Theta(\omega^2 \frac{\Delta}{\sqrt{k}}) 
\end{equation}

We next bound the last term:
\begin{align}
   \int_0^{\omega/4} dD \int_{\omega^*_b }^{\omega^*_b (1+\Delta/\sqrt{k})}\overline{\pi}(k x)dx &=
   \int_0^{\omega/4} dD \int_{\frac{9}{4}\omega+D }^{(\frac{9}{4}\omega+D)\cdot (1+\Delta/\sqrt{k})}\overline{\pi}(k x)dx\nonumber\\
   &\leq \int_0^{\omega/4} dD \int_{\frac{9}{4}\omega+D }^{(\frac{9}{4}\omega+D) + \frac{5}{2}\omega \cdot \frac{\Delta}{\sqrt{k}}}\overline{\pi}(k x)dx\nonumber\\
   &= \int_{0}^{\frac{5}{2}\omega \cdot \frac{\Delta}{\sqrt{k}}} dW \int_{\frac{9}{4}\omega +W}^{\frac{9}{4}\omega +W+\omega/4}  \overline{\pi}(k x)dx \nonumber\\
   &\leq \frac{5}{2}\omega \frac{\Delta}{\sqrt{k}}\cdot \frac{\omega}{4}\xi =\frac{5}{8}\frac{\Delta}{\sqrt{k}}\omega^2 \xi \;\; \text{(Apply Claim~\ref{Tforcorrectmap:claim})}\label{rightendbound:eq}
\end{align}   

We substitute \eqref{leftendbound:eq} and \eqref{rightendbound:eq} in \eqref{claimdiff:eq} to conclude the proof, choosing a small enough constant $\xi$.
\end{proof}

\begin{proof}[Proof of Lemma~\ref{scoregap:lemma}]
  
 We express the difference between the average score probability of a key in $N_0$ \eqref{scoreprobM:eq} and the score probability of a key in $N'$ \eqref{scoreprobGen:eq}:

   \begin{align}
   p_0(\pi)-p'(\pi) &=
     \frac{\sqrt{k}}{|N_0|} \cdot \int_{q_a}^{q_b} \int_{-\sqrt{k}}^\infty \overline{\pi}(\frac{k+\sqrt{k}z}{q}) f_Z(k; z) \cdot z d z \cdot f_\lambda(q) dq \nonumber\\ &=
     \frac{\sqrt{k}}{|N_0|} \cdot  \int_{-\sqrt{k}}^\infty \left( \int_{q_a}^{q_b} \overline{\pi}(\frac{k+\sqrt{k}z}{q})\cdot f_\lambda(q) dq \right) \cdot f_Z(k; z) \cdot z d z \label{scoreprobDiff:eq}
 \end{align}

We separately consider $Z$ in the range $I_{\text{in}}:= [-\sqrt{k}/4,\sqrt{k}/4]$ and 
$Z$ outside this range in $I_{\text{out}}=[-\sqrt{k},\sqrt{k}/4]\cup[\sqrt{k}/4,\infty]$

For outside the range we use that
$\int_{q_a}^{q_b} \overline{\pi}(\frac{k+\sqrt{k}z}{q})\cdot f_\lambda(q) dq \in [0,1]$ and tail bounds on $f_Z(k; z)$ \eqref{uptailZ:eq} \eqref{lowtailZ:eq} and get:
\begin{align}
   \frac{\sqrt{k}}{|N_0|}\int_{I_{\text{out}}} \left( \int_{q_a}^{q_b} \overline{\pi}(\frac{k+\sqrt{k}z}{q})\cdot f_\lambda(q) dq \right) \cdot f_Z(k; z) \cdot z d z = \frac{1}{N_0} e^{-\Omega(k)} \label{outrange:eq}\;\; \text{Apply \eqref{lowtailZ:eq} and \eqref{uptailZ:eq}}
\end{align}

For inside the range we apply Claim~\ref{deltascore:claim}:
\begin{align}
\lefteqn{\frac{\sqrt{k}}{|N_0|} \cdot  \int_{I_{\text{in}}} \left( \int_{q_a}^{q_b} \overline{\pi}(\frac{k+\sqrt{k}z}{q})\cdot f_\lambda(q) dq \right) \cdot f_Z(k; z) \cdot z d z} \nonumber\\
&=
\frac{\sqrt{k}}{|N_0|} \cdot \left( \int_{I_{\text{in}}} \left( \int_{q_a}^{q_b} \overline{\pi}(\frac{k}{q})\cdot f_\lambda(q) dq \right) \cdot f_Z(k; z) \cdot z d z + 
  \int_{I_{\text{in}}} \Theta(\frac{z}{\sqrt{k}}) \cdot f_Z(k; z) \cdot z d z \right) & \text{Claim~\ref{deltascore:claim}}
\nonumber\\
&=   \frac{\sqrt{k}}{|N_0|} \cdot \left( \int_{q_a}^{q_b} \overline{\pi}(\frac{k}{q})\cdot f_\lambda(q) dq \cdot \int_{I_{\text{in}}} f_Z(k; z) \cdot z d z + 
\frac{1}{\sqrt{k}}\cdot\Theta\left(\int_{I_{\text{in}}} f_Z(k; z) \cdot z^2 d z \right)\right)  \nonumber\\
&=   \frac{\sqrt{k}}{|N_0|} \cdot \left( \int_{q_a}^{q_b} \overline{\pi}(\frac{k}{q})\cdot f_\lambda(q) dq \cdot 0 + 
\frac{1}{\sqrt{k}}\cdot\Theta\left(1 \right)\right)  & \text{Using \eqref{EDelta:eq} and \eqref{squaredin:eq}}\nonumber\\
&=\frac{\sqrt{k}}{|N_0|}\cdot \frac{1}{\sqrt{k}}\cdot\Theta(1)  = \Theta(\frac{1}{|N_0|}) \label{inrange:eq}
\end{align}
The statement of the Lemma follows by combining \eqref{outrange:eq} and \eqref{inrange:eq} in~\eqref{scoreprobDiff:eq}.
\end{proof}

\subsection{The case of symmetric estimators} \label{symmetriccase:sec}
The proof of Lemma~\ref{symmetricsep:lemma} (gap for symmetric estimators) follows as a corollary of the proof of Lemma~\ref{scoregap:lemma}. 

\begin{proofof}{Lemma~\ref{symmetricsep:lemma}}
For symmetric maps (Definition~\ref{symmetric:def})  keys in $N_0$ that have lower rank can only have higher scoring probabilities. That is, when $j<j'$, the score probability of $m^i_j$ is no lower than that of $m^i_{j'}$. With bottom-$k$ sketches, the score probability of $m_j$ is no lower than that of $m_{j'}$.
In particular, the keys in $N^*_0$ have the highest average score among keys in their components.
Additionally, there is symmetry between components. Therefore, the average score of each of the $k$ lowest rank keys in $N^*_0$ is no lower than the average over all $N_0$ keys: 
\[
\E_{U}[\pi(S_\rho(U) \cdot \mathbf{1}(m\in U)] \geq p_0(\pi)\ .
\]
Therefore using Lemma~\ref{scoregap:lemma}:
    \[
\E_{U}[\pi(S_\rho(U) \cdot \mathbf{1}(m\in U)] -p'(\pi) \geq p_0(\pi) -p'(\pi)  - \E_{U}[\pi(S_\rho(U)\cdot \mathbf{1}(u\in U)] = \Omega(\frac{1}{k \log(kr)})\ .
\]
\end{proofof}

\subsection{Analysis details of the Adaptive Algorithm} \label{adaptivegendetails:sec}

We consider the information available to the query response algorithm.  The mask $M$ is shared with the query response and hence it only needs to estimate the cardinality of (the much larger) set $U$.  The mask keys $M$ hide information in $S_\rho(U)$ and make additional keys trasparent.  For $k$-partition and $k$-mins sketches, keys $m^i_h$ where $h> \arg\min_j m^i_j\in M$ are transparent.  For bottom-$k$ sketches, we see only $k'\leq k$ bottom ranks in $U$ if $k-k'$ keys from $M$ have lower ranks.

We describe the sampling of $S(U\cup M)$ (for given mask $M$) as an equivalent process that examines keys in $U$ in order, selecting each examined key with probability $q$, until the sketch is determined.
This process generalizes the process we described for the case without a mask in the proof of Lemma~\ref{dist:lemma}:

\begin{itemize}
\item
Bottom-$k$ sketch: Set counter $c\gets k$. $t\sim \Geom[q]$. Process keys $m_j$ in increasing $j$:
\begin{enumerate}
\item
    If $m_j\in M$ decrease $c$ and output $m_j$. If $c=0$ halt. 
\item
    If $m_j\not\in M$ then decrease $t$. 
\begin{enumerate}
    \item If $t=0$ output $m_j$, decrease $c$, and sample a new $t\sim \Geom[q]$. If $c=0$ halt.
\end{enumerate}      
\end{enumerate}
\item
$k$-mins and $k$-partition sketches:  
For $i\in[k]$ let
$h^i \gets \arg\min_\ell m^i_\ell \in M$.
Sample $t\sim \Geom[q]$. Process $i\in [k]$ in order: 
\begin{enumerate}
    \item If $h$ is defined and $h\leq t$ then $t\gets t-h+1$ and continue with next $i$. 
     \item
    If $h$ is undefined or $h>t$ then output $m^i_t$, sample new $t\sim \Geom[q]$. Continue to next $i$.
\end{enumerate}
\end{itemize}

The QR algorithm has the results of the process which yields $k' \leq k$ i.i.d.\ $\Geom[q]$ random variables. 
As keys are added to the mask $M$ the information we can glean on $q$ from the sketch, that corresponds to the number $k'$ of $\Geom[q]$ samples we obtain, decreases.  
As the mask gets augmented, the number of keys, additional keys in $N_0$ become transparent in the  sense that they have probability smaller than $\delta/r$ to impact the sketch if included in $U$.
With $k$-mins and $k$-partition sketches keys where $m^i_j > h^i$ become transparent. With bottom-$k$ sketches keys $m_j$ where $j> \min_i (k-|M\cap (m_\ell)_{\ell<j}|) \cdot \Omega(\log(r))$ are transparent.
These keys are no longer candidates to be examined by the process above.
We denote by $N'_0 \subset N_0$ the set of keys that remain non-transparent.  It holds that  
$|N'_0| = O(\overline{k'} \log(kr)$, where $\overline{k'}$ is the mean $k'$ with our current mask.
When $k' = O(\log(kr))$ becomes too small (see Remark~\ref{softthreshold:rem}), there are no correct maps and the algorithm halts and returns $M$.

Let $p(\pi,M,x)$ be the probability that key $x$ is scored with map $\pi$ and mask $M$.
The probability is the same for all transparent keys  $x\not\in N'_0$ and we denote it by $p'(\pi,M)$.
\begin{lemma} [Score probability gap with mask] \label{scoregapmask:lemma}
Let $\pi$ be a correct map for $M\cup\mathcal{D}_0$.
Then
\[
\sum_{x\in N'_0} \left( p(\pi,M,x) - p'(\pi,M) \right) = \Omega\left(\frac{1}{\log(kr)}\right)\ .
\]
\end{lemma}
\begin{proof}
    The proof is similar to that of Lemma~\ref{scoregap:lemma} applies with respect to $k'$ and using that $|N'_0| = O(\overline{k'} \log(kr)$.
\end{proof}

\begin{claim} \label{notransparentinM:claim}
  With probability at least $0.99$, no transparent keys are placed in $M$.
\end{claim}
\begin{proof}
First note that all transparent keys have the same score distribution (see proof of Theorem~\ref{onebatchutility:thm}).
Keys get placed in $M$ when their score 
separates from the median score in $N\setminus M$. Note that since nearly all keys (except $\alpha$ fraction) are transparent, the median score is the score of a transparent key.
From Chernoff bounds~\eqref{Chernoff:eq} the probability that a transparent key at a given step is placed in $M$ (and deviates by more than $\lambda$ from its expectation) is $< 1/(100 nr)$. Taking a union bound over all steps and transparent keys we obtain the claim.
\end{proof}

\end{document}